\documentclass[10pt,journal,twocolumn,twoside]{IEEEtran}

\usepackage{graphics} % for pdf, bitmapped graphics files
\usepackage{graphicx}
\usepackage{epsfig} % for postscript graphics files
\usepackage{amsmath} % assumes amsmath package installed
\usepackage{amssymb}  % assumes amsmath package installed

\usepackage{algpseudocode}
\usepackage{algorithm}

\usepackage{epstopdf}

\usepackage{lipsum}% http://ctan.org/pkg/lipsum
\usepackage{verbatim}

\usepackage{amsthm}

\usepackage{mathtools}
\usepackage{color}

\usepackage{hyperref}

\newtheorem{assumption}{\bf Assumption}
\newtheorem{definition}{\bf Definition}
\newtheorem{theorem}{\bf Theorem}
\newtheorem{proposition}{\bf Proposition}

\newtheorem{corollary}{\bf Corollary}
\newtheorem{remark}{\bf Remark}

\title{Robust output feedback model predictive control using online estimation bounds}

\author{Johannes K\"ohler$^{1,2}$, Matthias A. M\"uller$^3$, Frank Allg\"ower$^1$
\thanks{$^1$Institute for Systems Theory and Automatic Control, University of Stuttgart,
70550 Stuttgart, Germany. (email:allgower@ist.uni-stuttgart.de).}
\thanks{$^2$Institute for Dynamic Systems and Control, ETH Zürich, Switzerland. (email: jkoehle@ethz.ch)}
\thanks{$^3$Institute of Automatic Control, Leibniz University Hannover, 30167 Hannover, Germany.
(email:mueller@irt.uni-hannover.de).}
\thanks{%
This work was supported by the German Research Foundation (DFG) under Grant MU 3929/2-1.}}
\begin{document}
\maketitle
\thispagestyle{empty}
\pagestyle{empty}

%%%%%%%%%%%%%%%%%%%%%%%%%%%%%%%%%%%%%%%%%%%%%%%%%%%%%%%%%%%%%%%%%%%%%%%%%%%%%%%%

\begin{abstract}
We present a framework to design nonlinear robust output feedback model predictive control (MPC) schemes that ensure constraint satisfaction under noisy output measurements and disturbances. 
We provide novel estimation methods to bound the magnitude of the estimation error based on: stability properties of the observer; detectability; set-membership estimation; moving horizon estimation (MHE). 
Robust constraint satisfaction is guaranteed by suitably incorporating these online validated bounds on the estimation error in a homothetic tube based MPC formulation. 
In addition, we show how the performance can be further improved by combining MHE and MPC in a single optimization problem. 
The framework is applicable to a general class of detectable and (incrementally) stabilizable nonlinear systems.  
While standard output feedback MPC schemes use offline computed worst-case bounds on the estimation error, the proposed framework utilizes online validated bounds, thus reducing conservatism and improving performance. 
We demonstrate the reduced conservatism of the proposed framework using a nonlinear 10-state quadrotor example. 
\end{abstract}
%\begin{IEEEkeywords}
%Predictive control for nonlinear systems; Output feedback; Incremental system properties; Constrained control; Output regulation; Moving horizon estimation
%\end{IEEEkeywords}
%
%!TEX root = ./Output.tex
%%%%%%%%%%%%%%%%%%%%%%%%%%%%%%%%%%%%%%%%%%%%%%%%%%%%%%%%%%%%%%%%%%%%%%%%%%%%%%%
\section{Introduction}
Output feedback for nonlinear constrained systems is a theoretically challenging problem with high practical relevance. 
One of the main theoretical challenges in this problem setup includes the guaranteed satisfaction of safety relevant constraints despite the presence of uncertainty in terms of disturbances and state estimation error. 
%Examples
Examples of particular high interest include motion planning with vision based measurements (e.g., robotics and autonomous driving), where collision avoidance needs to be guaranteed despite potentially large uncertainty in the state estimate. 
In addition, solving the constrained output feedback problem is a preliminary for the constrained output regulation problem~\cite{Koehler2020Regulation}, which includes offset-free tracking~\cite{dong2020homothetic_offset} as a special cases.  
In this paper, we present a model predictive control (MPC) approach to the nonlinear constrained output feedback problem that combines modern robust MPC methodologies with online estimation bounds.

\subsubsection*{Related work}
%What is MPC
MPC~\cite{rawlings2017model,grune2017nonlinear} is an optimization based control method that can ensure satisfaction of state and input constraints
for general nonlinear systems. 
The presence of disturbances or estimation error can cause feasibility issues and invalidate the nominal stability properties~\cite{grimm2004examples}. 
%Robust MPC
\textit{Robust} MPC formulations account for bounded prediction mismatch by suitably adjusting a back off in the constraints and mitigate the effect of disturbances using an additional feedback.
Given a constant bound on the model mismatch, there exist many recent nonlinear robust MPC schemes that can ensure constraint satisfaction  using robust positive invariant (RPI) sets~\cite{singh2019robust,bayer2013discrete} or suitable over-approximations of the reachable set~\cite{limon2005robust,villanueva2017robust,kohler2018novel,Koehler2020Robust}. 

%Output feedback MPC
In the nonlinear output feedback case, the uncertain initial state estimate $\hat{x}$ and the lack of a separation principle further complicate the analysis and design of suitable MPC schemes, compare~\cite{findeisen2003state} for an overview. 
In the absence of state constraints, a nominal MPC implementation combined with a suitable observer can ensure (practical) asymptotic stability~\cite{findeisen2003output,magni2004stabilization,messina2005discrete}.
In the presence of state constraints, the combination of a robust MPC and a stable observer yields an output feedback MPC with some non-vanishing robustness margin in some (potentially small) region of attraction~\cite{roset2008robustness}. 
This result is, however, only of a qualitative nature and tailored robust output feedback MPC formulations can be significantly less conservative.

For linear systems, a joint minimax MHE and MPC optimization problem is proposed in~\cite{lofberg2002towards} resulting in  linear/quadratic matrix inequalities.
For nonlinear systems, in~\cite{copp2017simultaneous} a joint $\min-\max$ problem is proposed  that combines MPC and MHE.
However, this approach imposes additional conditions on the cost function to ensure a saddle point condition and the overall problem cannot be solved with standard solvers used in MPC.

%Tailored - output
The most straight forward method to design robust output feedback MPC schemes is to compute an RPI set offline that bounds the estimation error $x-\hat{x}\in\mathbb{E}$ and then use standard robust MPC methods. 
Corresponding linear output feedback MPC schemes have been developed using tubes~\cite{mayne2006robust,kogel2017robust}, constraint tightening~\cite{lovaas2008robust}, and more general feedback policies~\cite{goulart2007output,subramanian2016non}.
The special case of noisy state measurements can be handled analogously, which is particularly relevant for (data-driven) input-output models (cf.~\cite{manzano2020robust}).
The considered RPI set $\mathbb{E}$ can often be a conservative over-approximation of the true estimation error, which can lead to unnecessarily cautious control actions and thus lack of performance. 
The issue of a larger initial estimation error is addressed in~\cite{mayne2009robust} by pre-computing a sequence of monotonically decreasing sets $\mathbb{E}_t$, compare also~\cite{ji2020robust} for ellipsoidal sets. 
In~\cite{bemporad2000output}, set-membership estimation is used to compute a polytope $\mathbb{E}_{k|t}$ that contains the true state and is less conservative than offline computed RPI sets $\mathbb{E}$.
In~\cite{chisci2002feasibility}, the complexity and feasibility issues regarding the set-membership estimation method in~\cite{bemporad2000output} are illuminated and a solution is provided based on a fixed parametrization.  
Similar moving horizon estimation methods are used in~\cite{dong2020homothetic,dong2020homothetic_offset}  with a scalar parametrization $\gamma_k\mathbb{E}$ and a homothetic tube.
An alternative solution to the complexity and feasibility issue of the set-membership estimation is provided in~\cite{brunner2018enhancing} by using the last $M-k$ measurements to obtain valid set predictions $k$-steps in the future.

Overall, the existing  design procedures for robust output feedback MPC are tailored to linear system dynamics to efficiently compute reachable/invariants sets and and use polytopic set-membership estimation. 
Thus, the existing methods are not directly applicable to the nonlinear case.

\subsubsection*{Contribution}
In this work, we present a robust output feedback MPC framework for nonlinear constrained systems. 
As a first contribution, we develop estimation procedures to derive upper bounds for the observer error, which are applicable to nonlinear systems (cf. Section \ref{sec:estimation}). 
Then, as a second contribution, we develop a robust output-feedback MPC framework that utilizes online estimates of the magnitude of the observer error (cf. Section~\ref{sec:tube}).
By combining the novel estimation procedure with the robust MPC designs, we obtain nonlinear robust output feedback MPC formulations that ensure robust constraint satisfaction and can reduce the conservatism of offline bounds on the estimation error.
Overall, the resulting approach shares the main conceptual and theoretical properties of the linear output feedback MPC schemes~\cite{dong2020homothetic_offset,bemporad2000output,chisci2002feasibility,dong2020homothetic,brunner2018enhancing} and is applicable to a general class of nonlinear systems.

%estimation
In order to provide recursive feasibility guarantees in the MPC, the proposed estimation methods provide bounds for: the current estimation error, the future estimation error, and the deviations of the estimated state from the nominal dynamics.
First, we utilize detectability, i.e., incremental input/output-to-state stability ($\delta$-IOSS), in the form of equivalent dissipation inequalities (cf.~\cite{allan2020detect,knuefer2020time}) to determine valid bounds on the estimation error of a given stable observer using finite horizon past data. 
Furthermore, we show how this can be naturally extended to compute an ``optimal'' state estimate  (in terms of the derived bound), resulting in an MHE algorithm similar to~\cite{knufer2018robust,knuefer2021MHE}. 
In addition, we discuss how set-membership estimation methods based on the non-falsified set (cf.~\cite{dong2020homothetic_offset,bemporad2000output,chisci2002feasibility,dong2020homothetic,brunner2018enhancing})  can be applied to nonlinear systems.

%MPC
Given the derived bounds on the estimation error, we provide a general robust output-feedback MPC framework that allows for reduced conservatism based on the online computed bounds on the observer error.  
In particular, the proposed MPC framework uses incremental Lyapunov functions to derive a \textit{homothetic} tube formulation (cf.~\cite{dong2020homothetic,rakovic2012homothetic}). 
We also show how the MHE formulation can be integrated into the robust MPC to obtain an improved formulation that solves estimation and control in a single optimization problem, similar to~\cite{dong2020homothetic_offset,copp2017simultaneous,dong2020homothetic}. 
Finally, we show how the robust MPC formulation can be simplified to a nominal MPC formulation combined with a constraint tightening, thus allowing for efficient online implementation.  
A preliminary version of this formulation for the special case of offline computed bounds, exponential stability, and polytopic constraints can be found in the conference proceedings~\cite{Kohler2019Output}. 

Overall, the resulting framework is applicable to a large class of nonlinear systems, guarantees  recursive feasibility, constraint satisfaction, and robust performance bounds under bounded disturbances and noise. 
In contrast to most existing output-feedback MPC approaches (cf.~\cite{mayne2006robust,kogel2017robust,lee2001receding,lovaas2008robust,goulart2007output,mayne2009robust}), we are not restricted to offline computed worst-case bounds on the observer error, but use less conservative bounds validated during run-time, similar to~\cite{bemporad2000output,chisci2002feasibility,dong2020homothetic,brunner2018enhancing}. 
In the special case of exact state measurements, the proposed robust tube MPC formulation unifies earlier robust MPC schemes based on contraction metrics/incremental stability~\cite{singh2019robust,bayer2013discrete,kohler2018novel,Koehler2020Robust} by providing a homothetic tube formulation based on incremental Lyapunov functions. 
Due to the simple parametrization, the overall computational complexity of the proposed approach is only moderately increased compared to a nominal MPC.

\subsubsection*{Notation}
The quadratic norm with respect to a positive definite matrix $Q=Q^\top$ is denoted by $\|x\|_Q^2=x^\top Q x$ and  the minimal and maximal eigenvalue of $Q$ are denoted by $\lambda_{\min}(Q)$ and $\lambda_{\max}(Q)$, respectively. 
The identity matrix is denoted by $I_n\in\mathbb{R}^{n\times n}$. 
The non-negative real numbers are denoted by $\mathbb{R}_{\geq 0}$.
The set of integers is denoted by $\mathbb{I}$, $\mathbb{I}_{[a,b]}$ denotes the set of integers in the interval $[a,b]$ with some $a,b\in\mathbb{R}$, and $\mathbb{I}_{\geq 0}$ denotes the non-negative integer/natural numbers.
By $\mathcal{K}$ we denote the class of functions $\alpha:\mathbb{R}_{\geq 0}\rightarrow\mathbb{R}_{\geq 0}$, which are continuous, strictly increasing, and satisfy $\alpha(0)=0$. 
By $\mathcal{K}_{\infty}$ we denote the class of functions $\alpha\in\mathcal{K}$ which are unbounded.
We denote the class of functions $\delta:\mathbb{I}_{\geq 0}\rightarrow \mathbb{R}_{\geq 0}$, which are continuous and decreasing with $\lim\limits_{k\rightarrow\infty}\delta(k)=0$ by $\mathcal{L}$.
By $\mathcal{KL}$ we denote the functions $\beta:\mathbb{R}_{\geq 0}\times\mathbb{I}_{\geq 0}\rightarrow\mathbb{R}_{\geq 0}$ with $\beta(\cdot,t)\in\mathcal{K}$ and $\beta(r,\cdot)\in\mathcal{L}$ for any fixed $t\in\mathbb{I}_{\geq 0}$, $r\in\mathbb{R}_{\geq 0}$. 
The interior of a set $\mathbb{Z}\subset\mathbb{R}^n$ is denoted by $\text{int}(\mathbb{Z})$.

%!TEX root = ./Output.tex
%%%%%%%%%%%%%%%%%%%%%%%%%%%%%%%%%%%%%%%%%%%%%%%%%%%%%%%%%%%%%%%%%%%%%%%%%%%%%
\section{Preliminaries}
 \label{sec:setup_output}
This section introduces the problem setup,  the control goal, and preliminaries regarding detectability and stabilizability. 
%!TEX root = ./Output.tex
%%%%%%%%%%%%%%%%%%%%%%%%%%%%%%%%%%%%%%%%%%%%%%%%%%%%%%%%%%%%%%%%%%%%%%%%%%%%%
\subsection{Problem setup}
We consider a  nonlinear perturbed discrete-time system
\begin{subequations}
\label{eq:sys_w}
\begin{align}
\label{eq:sys_w_1}
x_{t+1}&=f_{\mathrm{w}}(x_t,u_t,w_t),\\
\label{eq:sys_w_2}
y_t&=h_{\mathrm{w}}(x_t,u_t,w_t), 
\end{align}
\end{subequations}
 with state $x_t\in\mathbb{X}=\mathbb{R}^n$, control input $u_t\in\mathbb{U}\subseteq\mathbb{R}^m$, disturbances/noise $w_t\in\mathbb{W}\subseteq\mathbb{R}^q$, noisy measurement $y_t\in\mathbb{Y}\subseteq\mathbb{R}^p$, time $t\in\mathbb{I}_{\geq 0}$, continuous dynamics $f_{\mathrm{w}}:\mathbb{X}\times\mathbb{U}\times\mathbb{W}\rightarrow\mathbb{X}$, and continuous output equations $h_{\mathrm{w}}:\mathbb{X}\times\mathbb{U}\times\mathbb{W}\rightarrow\mathbb{Y}$.  
We assume w.l.o.g. that $0\in\mathbb{W}$ and define the nominal system equations $f(x,u):=f_{\mathrm{w}}(x,u,0)$, $h(x,u):=h_{\mathrm{w}}(x,u,0)$. 
We impose point-wise in time constraints on the state and input
$(x_t,u_t)\in \mathbb{Z}$, $t\in\mathbb{I}_{\geq 0}$.
The overall control goal is to minimize some user chosen performance measure/cost $\ell$ while ensuring constraint satisfaction. 
To this end, we develop an output feedback MPC scheme that uses the past measured outputs $y_j$, $j\in\mathbb{I}_{[0,t-1]}$ and some initial state estimate $\hat{x}_0$ with a known bound on the estimation error to compute a control action $u_t$ at time $t$.

In order to derive robust bounds on the estimation error and ensure robust constraint satisfaction, we assume that the disturbances are bounded. 
 \begin{assumption}
\label{ass:disturbance} (Bounded disturbance)
There exists a constant $\overline{w}>0$, such that $\|w_t\|\leq \overline{w}$ for all $t\in\mathbb{I}_{\geq 0}$. 
\end{assumption}
   
%!TEX root = ./Output.tex
%%%%%%%%%%%%%%%%%%%%%%%%%%%%%%%%%%%%%%%%%%%%%%%%%%%%%%%%%%%%%%%%%%%%%%%%%%%%%
 \subsection{Detectability}
 \label{sec:setup_output_2}
 One standard description of detectability for nonlinear systems is the notion of incremental input/output-to-state stability ($\delta$-IOSS)~\cite{allan2020detect,knuefer2020time,knufer2018robust,knuefer2021MHE,sontag1997output,muller2017nonlinear,allan2019lyapunov}. 
\begin{definition}($\delta$-IOSS~\cite[Def.~3]{allan2020detect},\cite[Def.~2]{knuefer2020time})
\label{def:IOSS}
 System~\eqref{eq:sys_w} is incrementally uniformly IOSS if there exist $\beta_{1},\beta_{2},\beta_{3}\in\mathcal{KL}$ such that
\begin{align}
\label{eq:IOSS}
&\|x_{k}-\tilde{x}_{k}\|\leq \max\{\beta_{1}(\|x_0-\tilde{x}_0\|,k),\\
&\max_{j\in\mathbb{I}_{[0,k-1]}}\beta_{2}(\|w_j-\tilde{w}_j\|,k-j-1),\nonumber\\
&\max_{j\in\mathbb{I}_{[0,k-1]}}\beta_{3}(\|y_j-\tilde{y}_j\|,k-j-1)\},\nonumber
\end{align}
for all initial conditions $x,\tilde{x}\in\mathbb{X}$, disturbance sequences $w,\tilde{w}\in\mathbb{W}^\infty$, input sequences $u\in\mathbb{U}^{\infty}$, and all $k\in\mathbb{I}_{\geq 0}$, where $(x,u,y,w)_{t=0}^{\infty}$ and $(\tilde{x},u,\tilde{y},\tilde{w})_{t=0}^{\infty}$ correspond to two trajectories each satisfying Equations~\eqref{eq:sys_w} for all $t\in\mathbb{I}_{\geq 0}$. 
\end{definition}
We point out that this definition deviates from earlier characterizations used in the literature based on a maximum norm~\cite{sontag1997output,muller2017nonlinear,allan2019lyapunov} by using an explicit discounting in terms of $\mathcal{KL}$ functions. 
In particular, in~\cite{allan2020detect,knuefer2020time} it was recently shown that these characterizations are in fact equivalent (cf.~\cite[Prop.~4]{allan2020detect},  \cite[Thm.~5]{knuefer2020time}), this property (Def.~\ref{def:IOSS}) is necessary (cf.~\cite[Prop.~3]{knuefer2020time}, \cite[Prop.~5]{allan2020detect}) and sufficient (cf.~\cite[Thm.~13]{knuefer2021MHE}) for the design of robustly stable state observers, 
 and can be equivalently\footnote{%
We note that the considered output equation~\eqref{eq:sys_w_2} and the setting in~\cite{knuefer2020time} are more general compared to the setting in~\cite{allan2020detect} where $y=h(x)$ is assumed. We conjecture that the converse Lyapunov results  in~\cite{allan2020detect} remain valid. 
} characterized using a $\delta$-IOSS Lyapunov function  (cf.~\cite[Thm.~8]{allan2020detect}). 
\begin{definition} ($\delta$-IOSS Lyapunov function~\cite[Def.~6]{allan2020detect})
\label{def:IOSS_Lyap}
A function $W_\delta:\mathbb{X}\times\mathbb{X}\rightarrow\mathbb{R}_{\geq 0}$ is called an (exponential-decrease) $\delta$-IOSS Lyapunov function if there exist $\alpha_1,\alpha_2\in\mathcal{K}_{\infty}$, $\sigma_{1},\sigma_{2}\in\mathcal{K}$, and $\eta\in[0,1)$ such that
\begin{subequations}
\label{eq:IOSS_Lyap}
\begin{align}
\label{eq:IOSS_Lyap_1}
&\alpha_1(\|x-\tilde{x}\|)\leq W_\delta(x,\tilde{x})\leq \alpha_2(\|x-\tilde{x}\|),\\
\label{eq:IOSS_Lyap_2}
&W_\delta(f_{\mathrm{w}}(x,u,w),f_{\mathrm{w}}(\tilde{x},u,\tilde{w}))\leq \eta W_\delta(x,\tilde{x})\\
&+\sigma_{1}(\|w-\tilde{w}\|)+\sigma_{2}(\|h_{\mathrm{w}}(x,u,w)-h_{\mathrm{w}}(\tilde{x},u,\tilde{w}\|), \nonumber
\end{align}
\end{subequations}
for all $x,\tilde{x}\in\mathbb{X}$, $u\in\mathbb{U}$, $w,\tilde{w}\in\mathbb{W}$. 
\end{definition}
Note that $\delta$-IOSS is a special case of incremental dissipativity which under additional differentiability conditions and for $\alpha_1,\alpha_2$ quadratic can be equivalently characterized using the \textit{differential dynamics} (cf.~\cite{koelewijn2021incremental}) and thus corresponds to the \textit{differential detectability} conditions in~\cite{sanfelice2012metric}.  
In Section~\ref{sec:estimation}, we utilize this detectability characterizations to design observers and derive online verifiable bounds on the estimation error. 
  
%!TEX root = ./Output.tex
%%%%%%%%%%%%%%%%%%%%%%%%%%%%%%%%%%%%%%%%%%%%%%%%%%%%%%%%%%%%%%%%%%%%%%%%%%%%%
\subsection{Incremental stabilizability}
In order to ensure constraint satisfaction despite uncertain state estimates and disturbed dynamics, we require some analysis tool that allows us to efficiently compute an over-approximation of the reachable set for nonlinear systems. 
%increm
The stability of trajectories can be studied using the notion of \textit{incremental stability}~\cite{angeli2002lyapunov} and \textit{contraction metrics}~\cite{lohmiller1998contraction}, compare also~\cite{tran2019convergence}. 
%stabilizable
A less restrictive notion is given by incremental/universal \textit{stabilizability}~\cite{manchester2017control}, which considers an additional feedback. 
%ISS
 In order to also consider the presence of disturbances $w$, the notion of incremental \textit{stability} can be strengthened to incremental input-to-state stability ($\delta$-ISS)~\cite{bayer2013discrete,tran2016incremental}. 
A natural unification of these different concepts is to consider \textit{incremental input-to-state stabilizability}.
\begin{definition} (Incremental input-to-state stabilizability) 
\label{def:delta_ISS}
System~\eqref{eq:sys_w} is uniformly incremental input-to-state stabilizable, if there exists some control law $\kappa:\mathbb{X}\times\mathbb{X}\times\mathbb{U}\rightarrow\mathbb{U}$ and functions $\beta_{4},\beta_{5}\in\mathcal{KL}$, $\gamma_{\kappa}\in\mathcal{K}$, such that
\begin{subequations}
\label{eq:delta_ISS}
\begin{align}
\label{eq:delta_ISS_x}
\|x_{k}-\tilde{x}_{k}\|\leq &\max\{\beta_{4}(\|x_0-\tilde{x}_0\|,k),\\
&\max_{j\in\mathbb{I}_{[0,k-1]}}\beta_{5}(\|w_j-\tilde{w}_j\|,k-j-1)\},\nonumber\\
\label{eq:delta_ISS_kappa}
\|\kappa(\tilde{x}_k,{x}_k,{u_k})-u_k\|\leq &\gamma_{\kappa}(\|\tilde{x}_k-x_k\|),
\end{align}
\end{subequations}
for all initial conditions $x_0,\tilde{x}_0\in\mathbb{X}$, all disturbance sequences $w,\tilde{w}\in\mathbb{W}^\infty$, all input sequences $u\in\mathbb{U}^{\infty}$,
and all $k\in\mathbb{I}_{\geq 0}$, where $(x,u,y,w)_{t=0}^{\infty}$ and $(\tilde{x},\tilde{u},\tilde{y},\tilde{w})_{t=0}^{\infty}$ correspond to two different trajectories satisfying~\eqref{eq:sys_w} and $\tilde{u}_t=\kappa(\tilde{x}_t,x_t,u_t)$, $t\in\mathbb{I}_{\geq 0}$.  
\end{definition}
%
\begin{comment}
A natural unification of these different concepts is to consider \textit{attenuated} $\delta$-ISS (cf. attenuated ISS~\cite[Def.~4]{messina2005discrete}). 
\begin{definition} (Attenuated $\delta$-ISS) 
\label{def:delta_ISS}
System~\eqref{eq:sys_w} is attenuated uniformly $\delta$-ISS with some control law $\kappa:\mathbb{X}\times\mathbb{X}\times\mathbb{U}\rightarrow\mathbb{U}$ if there exist $\beta_{4},\beta_{5}\in\mathcal{KL}$, $\gamma_{\kappa}\in\mathcal{K}$, such that
\begin{subequations}
\label{eq:delta_ISS}
\begin{align}
\label{eq:delta_ISS_x}
\|x_{k}-\tilde{x}_{k}\|\leq &\beta_{4}(\|x_0-\tilde{x}_0\|,k)\\
&+\sum_{j=0}^{k-1}\beta_{5}(\|w_j-\tilde{w}_j\|,k-j-1),\nonumber\\
\label{eq:delta_ISS_kappa}
\|\kappa(\tilde{x}_k,{x}_k,{u_k})-u_k\|\leq &\gamma_{\kappa}(\|\tilde{x}_k-x_k\|),
\end{align}
\end{subequations}
for all initial conditions $x_0,\tilde{x}_0\in\mathbb{X}$, disturbance sequences $w,\tilde{w}\in\mathbb{W}^\infty$, input sequences $u\in\mathbb{U}^{\infty}$ and all $k\in\mathbb{I}_{\geq 0}$, where $(x,u,y,w)_{t=0}^{\infty}$ and $(\tilde{x},\tilde{u},\tilde{y},\tilde{w})_{t=0}^{\infty}$ correspond to two different trajectories satisfying~\eqref{eq:sys_w} and $\tilde{u}_t=\kappa(\tilde{x}_t,x_t,u_t)$. 
\end{definition}
\end{comment}
For $\kappa(\tilde{x},x,u)=u$ this definition reduces to $\delta$-ISS. 
The $\delta$-ISS characterizations in~\cite{bayer2013discrete,tran2016incremental} correspond to an alternative characterization using the $\max$-norm, analogous to the different characterizations of $\delta$-IOSS in~\cite{muller2017nonlinear,allan2019lyapunov} and \cite{allan2020detect,knuefer2020time} (compare the discussion below Def.~\ref{def:IOSS}).
In the absence of disturbances ($w=\tilde{w}=0$), this condition reduces to incremental stabilizability and condition~\eqref{eq:delta_ISS_kappa}  imposes additional (uniform) continuity bounds on $\kappa$. 
If additionally $\beta_{4},\beta_{5}$ are linear in the first argument and decay exponentially in the second argument, then this condition reduces to the \textit{universal/incremental exponential stabilizability} condition in~\cite{Koehler2020Robust,manchester2017control}. 
%Lyap
Since $\delta$-ISS is a special case of $\delta$-IOSS (e.g., with $h_{\mathrm{w}}$ constant), we use a control Lyapunov function (CLF) characterization analogous to Definition~\ref{def:IOSS_Lyap}.
\begin{definition} ($\delta$-ISS CLF)
\label{def:ISS_Lyap}
A function $V_\delta:\mathbb{X}\times\mathbb{X}\rightarrow\mathbb{R}_{\geq 0}$ is called an (exponential-decrease) $\delta$-ISS CLF if there exist $\alpha_3,\alpha_4\in\mathcal{K}_{\infty}$, $\sigma_{3},\gamma_{\kappa}\in\mathcal{K}$, $\rho\in[0,1)$, and a control law $\kappa:\mathbb{X}\times\mathbb{X}\times\mathbb{U}\rightarrow\mathbb{U}$ such that
\begin{subequations}
\label{eq:ISS_Lyap}
\begin{align}
\label{eq:ISS_Lyap_1}
&\alpha_3(\|x-\tilde{x}\|)\leq V_\delta(x,\tilde{x})\leq \alpha_4(\|x-\tilde{x}\|),\\
\label{eq:ISS_Lyap_2}
&V_\delta(f_{\mathrm{w}}(x,u,w),f_{\mathrm{w}}(\tilde{x},\kappa(\tilde{x},x,u),\tilde{w}))\\
\leq& \rho V_\delta(x,\tilde{x})+\sigma_{3}(\|w-\tilde{w}\|), \nonumber\\
\label{eq:ISS_Lyap_3}
&\|\kappa(\tilde{x},x,u)-u\|\leq \gamma_{\kappa}(\|x-\tilde{x}\|),
\end{align}
\end{subequations}
for all $x,\tilde{x}\in\mathbb{X}$, $u\in\mathbb{U}$, $w,\tilde{w}\in\mathbb{W}$. 
\end{definition}
Compared to the $\delta$-ISS Lyapunov characterization in~\cite[Def.~3]{bayer2013discrete}, \cite[Def.~7]{tran2016incremental} we consider an additional stabilizing feedback $\kappa$ and w.l.og.  (cf.~\cite[Prop.~7]{sontag1998comments}) use an exponential decrease characterization with $\rho\in[0,1)$. 
%Prop.~27 shows that global asymptotic stability can be equivalently characterized using an ``exponential-decay'' Lyapunov function
%Definition 7:tran: ``dissipation-form 
In the special case that $\alpha_3,\alpha_4$ are quadratic, $\gamma_{\kappa}$ is linear, and the dynamics are continuously differentiable, such an incremental CLF can be characterized using control contraction metrics (CCMs) (cf.~\cite{manchester2017control}), compare~\cite[Thm.~3.5]{singh2019robust}.  
In this case, there exist various recent offline design methods for $V_\delta,\kappa$ using sum-of-squares (SOS) programming~\cite{manchester2017control,singh2019robust} or linear matrix inequalities (LMIs) based on a parametrization using linear parameter varying (LPV) systems~\cite{koehler2020nonlinearTAC,wang2020virtual,koelewijn2021nonlinear}.
Similar methods can be used to compute a $\delta$-IOSS Lyapunov function~\cite{koelewijn2021incremental} or design a stable observer~\cite{sanfelice2012metric,manchester2020observer}.

In Section~\ref{sec:tube}, we utilize this incremental Lyapunov function to guarantee robust constraint satisfaction, similar to~\cite{singh2019robust,bayer2013discrete,kohler2018novel,Koehler2020Robust,Kohler2019Output}. % provide an over-approximation of the reachable set and thus 

\begin{remark}
\label{rk:constraints}
(Global bounds)
In this paper, we only consider \textit{global} system properties that are valid for all $(x,u,w)\in\mathbb{X}\times\mathbb{U}\times\mathbb{W}$, to simplify the exposition. 
However, it is possible to modify the following derivations in case the system properties only hold on the constraint set $\mathbb{Z}$ or to account for the fact that the stability properties of the observer may only hold for a small enough initial estimation error, compare~\cite{Kohler2019Output}. 
In addition, while this paper considers the standard case of constantly bounded disturbances (Ass.~\ref{ass:disturbance}), it is possible to incorporate state and input dependent bounds on the magnitude of the disturbances using the tools similar to~\cite{Koehler2020Robust}, in order to better reflect parametric model mismatch.
\end{remark}

%!TEX root = ./Output.tex
%%%%%%%%%%%%%%%%%%%%%%%%%%%%%%%%%%%%%%%%%%%%%%%%%%%%%%%%%%%%%%%%%%%%%%%%%%%%%%%
\section{State estimation bounds}
\label{sec:estimation}
In this section, we introduce the conditions on the state estimation and provide design procedures for corresponding state estimation methods. 
In order to ensure closed-loop constraint satisfaction and robust recursive feasibility with the proposed output-feedback MPC framework, the state observer should provide state estimates $\hat{x}_t$ satisfying the following properties:
\begin{enumerate}
\item The estimation error $\hat{x}_t-x_t$ satisfies a known bound $\overline{e}_t$.
\item The deviation of the observer dynamics $\hat{f}$ from the nominal dynamics $f$ satisfies a known bound.
\item An upper bound on the future estimation error is available. 
\end{enumerate}
These conditions are intuitively required as will also become apparent in Section~\ref{sec:tube}. 
In addition to these requirements, we are particularly interested in obtaining bounds on the estimation error, which are (significantly) less conservative than offline derived bounds. 
To this end, in Section~\ref{sec:estimation_luenberger} we will first investigate Luenberger-like observers. 
Then, in Sections~\ref{sec:IOSS} and \ref{sec:estimation_setmember} we provide online estimates on the observer error given the past data using detectability (Def.~\ref{def:IOSS_Lyap}) and set-membership estimation, respectively. 
Finally,  in Section~\ref{sec:estimation_MHE} we consider ``optimal'' state estimates using MHE. 

In order to simplify the following discussion, we will sometimes require  that the control input ensures boundedness of the closed loop, which is later guaranteed in Section~\ref{sec:tube} with the robust MPC design and compact constraints.
\begin{assumption}
\label{ass:boundedness} (Bounded trajectories)
There exists a constant $c>0$, such that for all $t\in\mathbb{I}_{\geq 0}$: 
$\max\{\|x_t\|,\|\hat{x}_t\|,\|u_t\|,\|w_t\|\}\leq c$.
\end{assumption}

%Lunenberg
%!TEX root = ./Output.tex
%%%%%%%%%%%%%%%%%%%%%%%%%%%%%%%%%%%%%%%%%%%%%%%%%%%%%%%%%%%%%%%%%%%%%%%%%%%%%%%
  \subsection{Luenberger-like observers}
\label{sec:estimation_luenberger}
The most standard observer design for linear and nonlinear systems is to use a copy of the nominal dynamics in combination with an injection law based on the measured output, which we refer to as Luenberger-like observers. 
The corresponding observer dynamics are given by\footnote{%
In case $h(x,u)$ is independent of $u$ we can also consider predictor-corrector observers that use $y_{t+1}$ instead of $y_t$ to compute $\hat{x}_{t+1}$.  
}
\begin{align}
\label{eq:observer}
\hat{x}_{t+1}=f(\hat{x}_t,u_t)+\hat{L}(\hat{x}_t,u_t,y_t)=:\hat{f}(\hat{x}_t,u_t,y_t),
\end{align}
with the state estimate $\hat{x}_t\in\mathbb{R}^n$ and the continuous injection law $\hat{L}:\mathbb{X}\times\mathbb{U}\times\mathbb{Y}\rightarrow\mathbb{X}$ with $\hat{L}(\hat{x},u,h(\hat{x},u))=0$. 
 \begin{assumption}
\label{ass:stable_observer} (Robustly stable observer)
There exist a $\delta$-Lyapunov function $V_{\mathrm{o}}:\mathbb{X}\times\mathbb{X}\rightarrow\mathbb{R}_{\geq 0}$ and $\alpha_5,\alpha_6\in\mathcal{K}_\infty$, $\sigma_4,\gamma_{\mathrm{L},1},\gamma_{\mathrm{L},2}\in\mathcal{K}$, $\tilde{\eta}\in[0,1)$, such that 
\begin{subequations}
\label{eq:observer_prop}
\begin{align}
\label{eq:observer_prop_1}
&\alpha_5(\|x-\hat{x}\|)\leq V_{\mathrm{o}}(\hat{x},x)\leq \alpha_6(\|x-\hat{x}\|)\\
\label{eq:observer_prop_2}
&V_{\mathrm{o}}(\hat{f}(\hat{x},u,h_{\mathrm{w}}(x,u,w)),f_{\mathrm{w}}(x,u,w))\\
\leq &\tilde{\eta} V_{\mathrm{o}}(\hat{x},x)+\sigma_4(\|w\|),\nonumber\\
\label{eq:observer_prop_3_mod}
&\|\hat{L}(\hat{x},u,h_{\mathrm{w}}(x,u,w))\|\leq \gamma_{\mathrm{L},1}(V_{\mathrm{o}}(\hat{x},x))+\gamma_{\mathrm{L},2}(\|w\|),
\end{align}
\end{subequations}
for all $(x,\hat{x},u,w)\in\mathbb{X}\times\mathbb{X}\times\mathbb{U}\times\mathbb{W}$.
 \end{assumption}
Conditions~\eqref{eq:observer_prop_1}-\eqref{eq:observer_prop_2} provide a Lyapunov characterization to ensure that $\hat{f}$ is a robustly stable observer (cf.~\cite[Def.~2]{allan2020detect}). 
Condition~\eqref{eq:observer_prop_3_mod} can be ensured by using  $\hat{L},h_{\mathrm{w}}$ continuous, Inequality~\eqref{eq:observer_prop_1}, and boundedness of the trajectories (Ass.~\ref{ass:boundedness}).
%!TEX root = ./Output.tex
%%%%%%%%%%%%%%%%%%%%%%%%%%%%%%%%%%%%%%%%%%%%%%%%%%%%%%%%%%%%%%%%%%%%%%%%%%%%%%%
\subsubsection*{Observer designs}
In the following, we detail different observer designs from the literature that satisfy Assumption~\ref{ass:stable_observer}. 
If the dynamics have the special form 
\begin{align*}
f(x,u)=Ax+\gamma(u,y),\quad h(x,u)=Cx+Du, 
\end{align*}
 the injection $\hat{L}(\hat{x},u,y)=\gamma(u,y)+L(C\hat{x}+Du-y)$ with some $L\in\mathbb{R}^{n\times p}$ yields linear dynamics for the observer error $e_{\mathrm{o}}$ and thus, assuming $(A,C)$ detectable, Conditions~\eqref{eq:observer_prop_1}--\eqref{eq:observer_prop_2} can be satisfied with $V_{\mathrm{o}}(x,\hat{x})=\|x-\hat{x}\|_{P_{\mathrm{o}}}^2$, $\alpha_5,\alpha_6$ quadratic and Condition~\eqref{eq:observer_prop_3_mod} follows from continuity of $\gamma$, compare~\cite{krener1985nonlinear}. 

If the dynamics $f$ is in an observer normal-form, then a \textit{high-gain observer} of the form 
\begin{align}
\label{eq:obs_linear_gain}
\hat{L}(\hat{x},u,y)=L\cdot (\hat{y}-y),\quad \hat{y}=h(\hat{x},u),\quad L\in\mathbb{R}^{n\times p},
\end{align}
can ensure exponential stability with a quadratic Lyapunov function $V_{\mathrm{o}}$ if the dynamics are Lipschitz continuous and $L$ ensures a sufficiently fast decay~\cite{gauthier1992simple}.
Note that a system can be transformed into observer normal-form if the system is \textit{uniformly observable for any input}~\cite[Thm.~2]{gauthier1992simple}. 
We note that the fast decay may require a \textit{high gain} $L$ and thus the bound in Condition~\eqref{eq:observer_prop_3_mod} can deteriorate.

If we allow for time-varying functions $V_{\mathrm{o}},L$, then Assumption~\ref{ass:stable_observer} can also be (locally) satisfied with an \textit{extended Kalman filter} (EKF). 
In particular, assuming continuous differentiability of $f,h$ with suitable boundedness and observability conditions on the Jacobian and small enough disturbances/noise, one can show that the EKF locally satisfies Assumption~\ref{ass:stable_observer} with linear functions $\gamma_{\mathrm{L},1},\gamma_{\mathrm{L},2}$ and a time-varying quadratic function $V_{\mathrm{o}}$, compare~\cite[Thm.~3.1]{reif1999stochastic}.  
A similar observer design with corresponding (local) stability properties is given by the state dependent Riccati equation (SDRE) Kalman filter~\cite{jaganath2005sdre}. 
However, the a priori guaranteed bounds on the estimation error for EKF and SDRE tend to be conservative, which limits the applicability in robust output-feedback MPC with safety critical state constraints.

The design of time-invariant observers of the form~\eqref{eq:obs_linear_gain} with quadratic Lyapunov functions $V_{\mathrm{o}}$ can be accomplished by treating the nonlinearity as a suitably bounded uncertainty. 
Corresponding LMIs can be found in~\cite{accikmecse2011observers} and \cite{zemouche2013lmi} for slope-restricted nonlinearites and based on an LPV embedding, respectively.

It is possible to consider more general quadratically bounded Lyapunov functions $V_{\mathrm{o}}$ and linearly bounded functions $\hat{L}$ by using the concept of 
\textit{differential detectability}~\cite{sanfelice2012metric}, which is a differential version of the $\delta$-IOSS property (Def.~\ref{def:IOSS_Lyap}). 
A corresponding differential observer can be designed using LMI/SOS tools~\cite{manchester2020observer}, dual to the construction of control contraction metrics~\cite{manchester2017control}.
For the special case $h(x,u)=[I_p,0]x$, a \textit{globally} exponentially stable observer is designed in reduced coordinates in~\cite{manchester2020observer}. 
More recently, general (partially necessary and sufficient) design conditions for \textit{globally} exponentially stable observers have been derived in~\cite{sanfelice2021metric}, utilizing more general coordinate arguments.

%IOSS
%!TEX root = ./Output.tex
%%%%%%%%%%%%%%%%%%%%%%%%%%%%%%%%%%%%%%%%%%%%%%%%%%%%%%%%%%%%%%%%%%%%%%%%%%%%%%%
\subsection{Estimation error bounds using detectability}
\label{sec:IOSS}
In the following, we show how recent data can be used to determine bounds on the current estimation error, which may be less conservative  than offline verifiable a priori bounds resulting from Assumptions~\ref{ass:disturbance} and \ref{ass:stable_observer}.
To this end, we exploit the boundedness of the disturbances (Ass.~\ref{ass:disturbance}), detectability (Def.~\ref{def:IOSS}, \ref{def:IOSS_Lyap}), and robust stability of the observer (Ass.~\ref{ass:stable_observer}).
Additional results regarding robustness w.r.t. outlier noise and observable systems can be found in Appendix~\ref{app:IOSS}.
First, we derive bounds based on the detectability conditions and improved bounds for the special case that the $\delta$-Lyapunov function $V_{\mathrm{o}}$ (Ass.~\ref{ass:stable_observer}) is also a $\delta$-IOSS Lyapunov function (Def.~\ref{def:IOSS_Lyap}). 
%IOSS
%!TEX root = ./Output.tex
%%%%%%%%%%%%%%%%%%%%%%%%%%%%%%%%%%%%%%%%%%%%%%%%%%%%%%%%%%%%%%%%%%%%%%%%%%%%%%%
The following assumption regarding the nature of the disturbances is crucial to allow for simple estimates based on the observer.
\begin{assumption}(Additive disturbances)
\label{ass:add_disturbance}
The perturbed dynamic~\eqref{eq:sys_w_1} satisfies  $f_{\mathrm{w}}(x,u,w)=f(x,u)+ E_{\mathrm{x}} w$ with $q\geq n$, $\mathbb{W}=\mathbb{R}^q$, and a full rank matrix  $E_{\mathrm{x}}\in\mathbb{R}^{n\times q}$. 
\end{assumption}

\begin{proposition}
\label{prop:IOSS}
Let Assumptions~\ref{ass:disturbance} and \ref{ass:add_disturbance} hold.
Suppose the system admits an (exponential-decay) $\delta$-IOSS Lyapunov function (Def.~\ref{def:IOSS_Lyap}).
Then, for any $t\in\mathbb{I}_{\geq 0}$, $M\in\mathbb{I}_{[0,t]}$  the state estimate of the observer~\eqref{eq:observer} satisfies 
\begin{align}
\label{eq:observer_IOSS_bound_1}
&W_\delta(\hat{x}_t,x_t)\nonumber\\
\leq &\sum_{j=1}^{M}\eta^{j-1}\left(\sigma_1(\overline{w}+\|\hat{w}_{t-j}\|)+\sigma_2(\|\hat{y}_{t-j}-y_{t-j}\|) \right)\nonumber\\
&+\eta^M W_\delta(\hat{x}_{t-M},x_{t-M}),
\end{align}
with $\hat{y}_k=h_{\mathrm{w}}(\hat{x}_k,u_k,\hat{w}_k)$, $\hat{w}_k=E_{\mathrm{x}}^\dagger \hat{L}(\hat{x}_k,u_k,y_k)$, $k\in\mathbb{I}_{\geq 0}$, where $E_{\mathrm{x}}^\dagger=E_{\mathrm{x}}^\top (E_{\mathrm{x}} E_{\mathrm{x}}^\top)^{-1}$ is the  Moore–Penrose right inverse of $E_{\mathrm{x}}$.
\begin{comment}
If the system is uniformly observable (Def.~\ref{def:observable}), then for any $t\geq \nu$ the following inequality holds
\begin{align}
\label{eq:observer_IOSS_bound_2}
\|\hat{x}_{t-\nu}-{x}_{t-\nu}\|\leq \sum_{j=1}^{\nu}\gamma_{\mathrm{w}}(\overline{w}+\|\hat{w}_{t-j}\|)+\gamma_{\mathrm{v}}(\|y_{t-j}-\hat{y}_{t-j}\|).
\end{align}
\textbf{Part I: }
\end{comment}
\end{proposition}
\begin{proof}
Due to Assumption~\ref{ass:add_disturbance} and the observer structure in Equation~\eqref{eq:observer}, we have
\begin{align*}
\hat{x}_{k+1}=&\hat{f}(\hat{x}_k,u_k,y_k)=f(\hat{x}_k,u_k)+\hat{L}(\hat{x}_k,u_k,y_k)\\
=&f(\hat{x}_k,u_k)+E_{\mathrm{x}}\hat{w}_k = f_{\mathrm{w}}(\hat{x}_k,u_k,\hat{w}_k).
\end{align*}
Applying Inequality~\eqref{eq:IOSS_Lyap_2} for some $k\in\mathbb{I}_{\geq 0}$ yields
\begin{align*}
&W_{\delta}(\hat{x}_{k+1},x_{k+1})=W_{\delta}(f_{\mathrm{w}}(\hat{x}_k,u_k,\hat{w}_k),f_{\mathrm{w}}(x_k,u_k,w_k))\\
\leq& \eta W_{\delta}(\hat{x}_k,x_k)+\sigma_1(\|\hat{w}_k-{w}_k\|)\\
&+\sigma_2(\|h_{\mathrm{w}}(\hat{x}_k,u_k,\hat{w}_k)-h_{\mathrm{w}}(x_k,u_k,w_k)\|)\\
\stackrel{\text{Ass.}~\ref{ass:disturbance}}{\leq}& \eta W_{\delta}(\hat{x}_k,x_k)+\sigma_1(\overline{w}+\|\hat{w}_k\|)+\sigma_2(\|\hat{y}_k-y_k\|). 
\end{align*}
Using this inequality repeatedly for $k\in\mathbb{I}_{[t-M,t-1]}$, we obtain Inequality~\eqref{eq:observer_IOSS_bound_1}.
\begin{comment}
\textbf{Part II: } Analogous to the preceding proof, $\hat{x}_t$ is the solution to the perturbed dynamics $f_{\mathrm{w}}$ with the ``disturbances'' $\hat{w}_k$. 
Thus, Condition~\eqref{eq:observability} directly yields Inequality~\eqref{eq:observer_IOSS_bound_2}. 
\end{comment}
\end{proof}
Note that given given a bound on the  estimation error at time $t-M$, Inequality~\eqref{eq:observer_IOSS_bound_1} provides a valid bound on the current estimation error using the measured quantities $\hat{w}_k,\hat{y}_k,y_k$, the disturbance bound $\overline{w}$, and detectability (Def.~\ref{def:IOSS_Lyap}). 
In the extreme case that the observer would exactly match the data ($\hat{w}=0$, $\hat{y}=y$), the bound yields an exponential decay in terms of the initial estimation error and an additive term $\sigma_1(\overline{w})$. 
Given Definition~\ref{def:IOSS_Lyap}, the only additional conservatism of the derived bound is the inequality $\|\hat{w}_k-w_k\|\leq\|\hat{w}_k\|+\overline{w}$, which is needed since $w_k$ is unknown.  
\begin{remark}
\label{rk:add_dist}
(Additive disturbances)
The main restriction posed in Assumption~\ref{ass:add_disturbance} is that the dynamics are one-step controllable w.r.t. disturbances $w$, similar to~\cite{knufer2018robust}. 
In case the system is not controllable w.r.t. $w$, we can artificially introduce additional additive disturbances $v$ in the model. 
Then, Inequality~\eqref{eq:IOSS_Lyap_2} from $\delta$-IOSS can be adapted to
\begin{align*}
%\label{eq:IOSS_Lyap_2_vict}
&W_\delta(f_{\mathrm{w}}(x,u,w)+v,f_{\mathrm{w}}(\tilde{x},u,\tilde{w}))\leq \eta W_\delta(x,\tilde{x})+\sigma_{\delta}(\|v-0\|)\nonumber\\
&+\sigma_{1}(\|w-\tilde{w}\|)+\sigma_{2}(\|h_{\mathrm{w}}(x,u,w)-h_{\mathrm{w}}(\tilde{x},u,\tilde{w}\|),
\end{align*}
assuming $W_\delta$ is uniformly continuous (Ass.~\ref{ass:IOSS_cont}). 
Correspondingly, Inequality~\eqref{eq:observer_IOSS_bound_1} remains valid with $\sigma_1(\overline{w}+\|\hat{w}_{t-j}\|)$ replaced by $\sigma_1(\overline{w})+\sigma_\delta(\| \hat{v}_k\|)$ with $v_k=\hat{L}(\hat{x}_k,u_k,y_k)$. 
%Thus, we have a system which is one-step controllable w.r.t. $v$ and satisfies a $\delta$-IOSS inequality. 
Hence, for all intense and purposes, Assumption~\ref{ass:add_disturbance} can be relaxed by defining an additional artificial additive disturbance $v$, assuming (uniform) continuity of $W_\delta$.
\end{remark}

By combining the $\delta$-IOSS based estimate with the stability properties of the observer (Ass.~\ref{ass:stable_observer}), we can compute a scalar bound $\overline{e}_t$ on the estimation error.
In particular, given some initial bound $\overline{e}_0\geq 0$, we recursively use the update 
\begin{subequations}
\label{eq:e_t_bound_1}
 \begin{align}
\label{eq:e_t_bound_1_a}
\overline{e}_{t,\mathrm{IOSS},M}:=&\sum_{j=1}^M \eta^{j-1} (\sigma_1(\overline{w}+\|\hat{w}_{t-j}\|)+\sigma_2(\|\hat{y}_{t-j}-y_{t-j}\|)) \nonumber\\
&+\eta^M\alpha_2(\alpha_5^{-1}(\overline{e}_{t-M})), \quad M\in\mathbb{I}_{[1,\min\{t,\overline{M}\}]}, \\
\label{eq:e_t_bound_1_b}
\overline{e}_{t,\mathrm{IOSS}}:=&\alpha_6(\alpha_1^{-1}(\min_{M\in\mathbb{I}_{[1,\min\{t,\overline{M}\}]}}\overline{e}_{t,\mathrm{IOSS},M})), \\
\label{eq:e_t_bound_1_c}
\overline{e}_t:=&\min\{\tilde{\eta} \overline{e}_{t-1}+\sigma_4(\overline{w}), \overline{e}_{t,\mathrm{IOSS}}\},
\end{align}
\end{subequations}
with some $\overline{M}\in\mathbb{I}_{\geq 1}$ specified by the user.
\begin{theorem}
\label{thm:IOSS}
Let Assumptions~\ref{ass:disturbance}, \ref{ass:stable_observer} and \ref{ass:add_disturbance} hold.
Suppose the system admits an (exponential-decay) $\delta$-IOSS Lyapunov function (Def.~\ref{def:IOSS_Lyap}) and that $V_{\mathrm{o}}(\hat{x}_0,x_0)\leq \overline{e}_0$.
Then, for all $t\in\mathbb{I}_{\geq 0}$ the estimates $\hat{x}_t$, $\overline{e}_t$ according to~\eqref{eq:observer} and \eqref{eq:e_t_bound_1} satisfy
\begin{subequations}
\label{eq:thm_IOSS}
\begin{align}
\label{eq:thm_IOSS_1}
V_{\mathrm{o}}(\hat{x}_t,x_t)\leq &\overline{e}_t,\\
\label{eq:thm_IOSS_2}
\|\hat{x}_{t+1}-f(\hat{x}_t,u_t)\|\leq& \gamma_{\mathrm{L},1}(\overline{e}_t)+\gamma_{\mathrm{L},2}(\overline{w}),\\
\label{eq:thm_IOSS_3}
\overline{e}_{t+k}\leq& \eta^k\overline{e}_t+\dfrac{1-\eta^k}{1-\eta}\sigma_4(\overline{w}), ~ k\in\mathbb{I}_{\geq 0}. 
\end{align}
\end{subequations}
\end{theorem}
\begin{proof}
We first show Inequality~\eqref{eq:thm_IOSS_1} using a proof of induction. Suppose that $V_{\mathrm{o}}(\hat{x}_k,x_k)\leq \overline{e}_k$ $\forall k\in\mathbb{I}_{[0,t-1]}$. 
Inequalities~\eqref{eq:IOSS_Lyap_1} and \eqref{eq:observer_prop_1} imply
\begin{align}
\label{eq:proof_IOSS_1}
\alpha_1(\alpha_6^{-1}(V_{\mathrm{o}}(\hat{x},x)))\leq W_\delta(\hat{x},x)\leq \alpha_2(\alpha_5^{-1}(V_{\mathrm{o}}(\hat{x},x))).
\end{align}
Thus, Inequality~\eqref{eq:observer_IOSS_bound_1} ensures that $W_{\delta}(\hat{x}_t,x_t)\leq \min_M\overline{e}_{t,\mathrm{IOSS},M}$ and thus $V_{\mathrm{o}}(\hat{x}_t,x_t)\leq \alpha_6(\alpha_1^{-1}(\min_M\overline{e}_{t,\mathrm{IOSS},M}))=\overline{e}_{t,\mathrm{IOSS}}$.
 Condition~\eqref{eq:observer_prop_2} and Assumption~\ref{ass:disturbance} directly imply that $V_{\mathrm{o}}(\hat{x}_t,x_t)\leq \tilde{\eta}\overline{e}_{t-1}+\sigma_4(\overline{w})$.
 Thus, $V_{\mathrm{o}}(\hat{x}_t,x_t)\leq \overline{e}_t$. 
 Condition~\eqref{eq:thm_IOSS_2} follows directly from Inequality~\eqref{eq:observer_prop_3_mod}
using~\eqref{eq:thm_IOSS_1} and Assumption~\ref{ass:disturbance}. 
Condition~\eqref{eq:thm_IOSS_3} follows by applying the bound $\overline{e}_t\leq\tilde{\eta} \overline{e}_{t-1}+\sigma_4(\overline{w})$ from Equation~\eqref{eq:e_t_bound_1_c} $k$ times using the geometric series.
\end{proof}
Theorem~\ref{thm:IOSS} provides all the properties we required from the state estimation. 
In particular, the update rule~\eqref{eq:e_t_bound_1} yields valid bounds $\overline{e}_t$ (cf.~\eqref{eq:thm_IOSS_1}) on the estimation error, which also use recent data to potentially reduce conservatism. 
To allow for a reduction in conservatism, the ideal bound $\sigma_1(\overline{w})/(1-\eta)$ (perfectly matched data, $\hat{y}=0,\hat{w}=0$) should be smaller than the a priori observer bound $\sigma_4(\overline{w})/(1-\tilde{\eta})$. 
Condition~\eqref{eq:thm_IOSS_2} provides a bound on the difference between the nominal prediction model and the observer dynamics. 
Condition~\eqref{eq:thm_IOSS_3} in combination with Inequalities~\eqref{eq:thm_IOSS_1}--\eqref{eq:thm_IOSS_2}  allows us to predict valid bounds on the estimation error and the prediction mismatch. 
In the special case that we use no past data for the observer bounds ($\overline{M}=0$), we recover the simple error propagation used in our preliminary work~\cite{Kohler2019Output}. 
For $\overline{e}_0>\frac{\sigma_4(\overline{w})}{1-\tilde{\eta}}$, this is similar to the monotonically decreasing error sets used in~\cite{mayne2009robust} for linear systems.

%IOSS with common Lyap
%!TEX root = ./Output.tex
%%%%%%%%%%%%%%%%%%%%%%%%%%%%%%%%%%%%%%%%%%%%%%%%%%%%%%%%%%%%%%%%%%%%%%%%%%%%%%%
\subsubsection*{Identical Lyapunov function}
In the following, we investigate the important special case, when the $\delta$-Lyapunov function $V_{\mathrm{o}}$ (Ass.~\ref{ass:stable_observer}) is also a $\delta$-IOSS Lyapunov function.  
\begin{assumption} (Identical Lyapunov function)
\label{ass:common_Lyap}
The $\delta$-Lyapunov function $V_{\mathrm{o}}$ from Assumption~\ref{ass:stable_observer} is also an (exponential-decrease) $\delta$-IOSS Lyapunov function according to Definition~\ref{def:IOSS_Lyap}.  
\end{assumption}
This condition is  naturally satisfied if $f_{\mathrm{w}},h_{\mathrm{w}}$ are affine in $w$, $V_{\mathrm{o}}$ is quadratic and $\hat{L}$ according to~\eqref{eq:obs_linear_gain}, compare Appendix~\ref{app:identical}. 
Given $W_{\delta}=V_{\mathrm{o}}$ (Ass.~\ref{ass:common_Lyap}), the update rule~\eqref{eq:e_t_bound_1} simplifies to
\begin{subequations}
\label{eq:e_t_bound_2}
 \begin{align}
\label{eq:e_t_bound_2_a}
\overline{e}_{t,\mathrm{IOSS},M}:=&\sum_{j=1}^M \eta^{j-1} (\sigma_1(\overline{w}+\|\hat{w}_{t-j}\|)+\sigma_2(\|\hat{y}_{t-j}-y_{t-j}\|))\nonumber\\
&+\eta^M\overline{e}_{t-M}, \quad M\in\mathbb{I}_{[1,\min\{t,\overline{M}\}]}, \\
\label{eq:e_t_bound_2_b}
\overline{e}_{t,\mathrm{IOSS}}:=&\min_{M\in\mathbb{I}_{[1,\min\{t,\overline{M}\}]}}\overline{e}_{t,\mathrm{IOSS},M}, \\ 
\label{eq:e_t_bound_2_c}
\overline{e}_t:=&\min\{\tilde{\eta}\overline{e}_{t-1}+\sigma_4(\overline{w}), \overline{e}_{t,\mathrm{IOSS}}\}.
\end{align}
\end{subequations}
The following corollary shows that the results in Theorem~\ref{thm:IOSS} remain valid and that is suffices to set $\overline{M}=1$, i.e., only use the most recent measurement to compute $\overline{e}_t$. 
\begin{corollary}
\label{corol:IOSS}
Let Assumption~\ref{ass:common_Lyap} and the conditions in Theorem~\ref{thm:IOSS} hold.
Then, for any $t\in\mathbb{I}_{\geq 0}$, the estimates $\hat{x}_t$, $\overline{e}_t$ according to~\eqref{eq:observer} and \eqref{eq:e_t_bound_2} satisfy Inequalities~\eqref{eq:thm_IOSS}.
Furthermore, $\overline{e}_{t,\mathrm{IOSS}}=\overline{e}_{t,\mathrm{IOSS},1}$. 
\end{corollary}
\begin{proof}
The proof of Theorem~\ref{thm:IOSS} remains unchanged, except for Inequality~\eqref{eq:proof_IOSS_1}, which is replaced by identity due to Assumption~\ref{ass:common_Lyap}. 
Abbreviate $\sigma_t=\sigma_1(\overline{w}+\|\hat{w}_{t}\|)+\sigma_2(\|\hat{y}_{t}-y_{t}\|)$, $t\in\mathbb{I}_{\geq 0}$. 
For any $t\in\mathbb{I}_{\geq 0}$, $M\in\mathbb{I}_{[0,\min\{\overline{M},t\}-1]}$ we have
\begin{align*}
&\overline{e}_{t,\mathrm{IOSS},M+1}-\overline{e}_{t,\mathrm{IOSS},M}
=\eta^M(\sigma_{t-M-1}+\eta\overline{e}_{t-M-1}- \overline{e}_{t-M}).
\end{align*}
Given that 
\begin{align*}
\overline{e}_{t-M}\leq \overline{e}_{t-M,\mathrm{IOSS},1}=\eta\overline{e}_{t-M-1}+\sigma_{t-M-1},
\end{align*}
we get $\overline{e}_{t,\mathrm{IOSS},M}\leq \overline{e}_{t,\mathrm{IOSS},M+1}$.
Since this holds for all $M\in\mathbb{I}_{[0,\min\{\overline{M},t\}-1]}$, we have
$\overline{e}_{t,\mathrm{IOSS},1}\leq\overline{e}_{t,\mathrm{IOSS},M}$ for all $M\in\mathbb{I}_{[1,\min\{t,\overline{M}\}]}$ and thus $\overline{e}_{t,\mathrm{IOSS}}=\overline{e}_{t,\mathrm{IOSS},1}$. 
\end{proof}
Based on this result we can reduce the update rule to
 \begin{align}
\label{eq:e_t_bound_3}
\overline{e}_{t+1,\mathrm{IOSS}}:=&\eta\overline{e}_{t}+\sigma_1(\overline{w}+\|\hat{w}_{t}\|)+\sigma_2(\|\hat{y}_{t}-y_{t}\|)), \nonumber\\
\overline{e}_{t+1}:=&\min\{\tilde{\eta}\overline{e}_{t}+\sigma_4(\overline{w}), \overline{e}_{t+1,\mathrm{IOSS}}\},
\end{align}
which can be evaluated very efficiently. 
We point out that for the general case considered in Theorem~\ref{thm:IOSS}, a larger value of $\overline{M}$ is typically advantageous since the conservatism induced by the factor $\alpha_2\circ\alpha_5^{-1}$ vanishes for large $M$.

%outlier
%\input{Estimation_2_3}
%observable
%\input{Estimation_2_4}

%Set
%!TEX root = ./Output.tex
%%%%%%%%%%%%%%%%%%%%%%%%%%%%%%%%%%%%%%%%%%%%%%%%%%%%%%%%%%%%%%%%%%%%%%%%%%%%%%%
\subsection{Estimation error bounds using set-membership estimation}
\label{sec:estimation_setmember}
In the following, we discuss set-membership methods to compute bounds on the estimation error. 
First, we discuss the non-falsified set.
Then, we present an optimization-based estimate using a fixed scalar parametrization and a moving horizon estimation.

\subsubsection*{Non-falsified set}
A classical approach to compute the set of possible states $x_t$ given past measurements $(y_k,u_k)_{k=0}^{t-1}$ is the so called non-falsified set.
In particular, given some set $\mathbb{E}_t\subseteq\mathbb{X}$ with $x_{t}\in\mathbb{E}_t$, the disturbance bound (Ass.~\ref{ass:disturbance}) and the measured input and output $u_t,y_t$, the non-falsified set can be updated as 
$\mathbb{E}_{t+1}:=\mathcal{F}(\mathbb{E}_t,u_t,y_t)$
with the set-valued map
\begin{align}
\label{eq:set_propagate}
\mathcal{F}(\mathbb{E},u,y):=\{&f_{\mathrm{w}}(x,u,w)|~\exists (w,x)\in\mathbb{W}\times\mathbb{E}: \\
&\|w\|\leq \overline{w},~y=h_{\mathrm{w}}(x,u,w)\}.\nonumber
\end{align}
The resulting sets $\mathbb{E}_t$ are the smallest possible sets that are guaranteed to contain the true state $x_t$, given the prior assumptions and measurements.

%linear
In case of linear systems with polytopic disturbance bounds, polytopic sets $\mathbb{E}_t$ can be efficiently computed by stacking the corresponding inequality constraints~\cite{bemporad2000output}.
However, the complexity of the set $\mathbb{E}_t$ increases unboundedly during runtime and simply discarding old measurements may yield recursive feasibility issues in the MPC, compare~\cite{chisci2002feasibility}. 
This feasibility issue has been solved in~\cite{brunner2018enhancing} by using the past $M-k$ measurements to define a feasible set $\mathbb{E}_{k|t}$, which is used for the robust predictions $k$ steps into the future, $k\in\mathbb{I}_{[0,N]}$. 
In case of linear systems with ellipsoidal bounds, ellipsoidal sets $\mathbb{E}_t$ can be computed using the methods developed in~\cite{bertsekas1971recursive}.

\subsubsection*{Fixed-complexity block updates}
The complexity and feasibility issues associated with the non-falsified set can be solved by using a finite horizon window to compute a set $\mathbb{E}_t$ with a fixed parametrization  that over-approximates the non-falsified set. 
For linear systems with polytopic disturbance bounds, corresponding polytopes $\mathbb{E}_t$ can be computed using block recursive updates based on linear programs (LPs) (cf.~\cite{chisci2002feasibility}), which can be integrated in an MHE-MPC formulation (cf.~\cite{dong2020homothetic_offset,dong2020homothetic}). %,chisci1998block

We extend this idea to the nonlinear setting by considering sets of the form $\mathbb{E}_t:=\{x\in\mathbb{X}|~V_{\mathrm{o}}(\hat{x}_t,x)\leq \overline{e}_t\}$ which are centred around the Luenberger state estimate $\hat{x}_t$ with some variable scaling $\overline{e}_t\geq 0$. 
At time $t$, given the past $M_t=\min\{\overline{M},t\}$ measurements and some initial bound $\overline{e}_{t-M_t}\geq 0$, we solve the following nonlinear program (NLP)%:
\begin{subequations}
\label{eq:NLP_set_estimation}
\begin{align}
\hat{\gamma}_{t,M_t}:=&\max_{\overline{w}_{\cdot|t},\overline{x}_{-M_t|t}} V_{\mathrm{o}}(\hat{x}_t,\overline{x}_{0|t})\\
\label{eq:NLP_set_estimation_init}
\text{s.t. }& V_{\mathrm{o}}(\overline{x}_{-M_t|t},\hat{x}_{t-M_t})\leq \overline{e}_{t-M_t},\\
\label{eq:NLP_set_estimation_dyn}
&\overline{x}_{k+1|t}=f_{\mathrm{w}}(\overline{x}_{k|t},u_{t+k},\overline{w}_{k|t}),~ k\in\mathbb{I}_{[-M_t,-1]},\\
\label{eq:NLP_set_estimation_y}
&y_{t-k}=h_{\mathrm{w}}(\overline{x}_{k|t},u_{t+k},\overline{w}_{k|t}),~ k\in\mathbb{I}_{[-M_t,-1]},\\
\label{eq:NLP_set_estimation_w}
&\|\overline{w}_{k|t}\|\leq \overline{w},~ k\in\mathbb{I}_{[-M_t,-1]}.
\end{align}
\end{subequations}
A maximizer is denoted by $\overline{x}^*_{-M_t|t},\overline{w}^*_{\cdot|t}$. 
In the special case of linear dynamics $f_{\mathrm{w}},h_{\mathrm{w}}$, polytopic bounds on $w$, and a polytopic function $V_{\mathrm{o}}$, the optimization problem~\eqref{eq:NLP_set_estimation} reduces to an LP, similar to the updates used in~\cite{dong2020homothetic_offset,chisci2002feasibility,dong2020homothetic}.

In order to provide recursively feasible and predictable bounds on the magnitude of the observer error, we additionally\footnote{%
For $M=1$, Assumption~\ref{ass:stable_observer} intuitively ensures $\hat{\gamma}_{t,M_t}\leq \tilde{\eta}\bar{e}_{t-1}+\sigma_4(\bar{w})$.
} use the stability properties of the observer (Ass.~\ref{ass:stable_observer}) to define
\begin{align}
\label{eq:NLP_set_estimation_update}
\overline{e}_t:=\min\{\hat{\gamma}_{t,M_t},\tilde{\eta}\overline{e}_{t-1}+\sigma_4(\overline{w})\},
\end{align}
analogous to the update in Equation~\eqref{eq:e_t_bound_1_c}.  
\begin{theorem}
\label{thm:set}
Let Assumptions~\ref{ass:disturbance} and \ref{ass:stable_observer} hold.
Suppose that $V_{\mathrm{o}}(\hat{x}_0,x_0)\leq \overline{e}_0$.
Then, for all $t,\overline{M}\in\mathbb{I}_{\geq 0}$ the estimates $\hat{x}_t$, $\overline{e}_t$ according to~\eqref{eq:observer} and \eqref{eq:NLP_set_estimation}--\eqref{eq:NLP_set_estimation_update} satisfy Inequalities~\eqref{eq:thm_IOSS}. 
\end{theorem}
\begin{proof}
We first show Inequality~\eqref{eq:thm_IOSS_1} using a proof of induction. 
Suppose that $V_{\mathrm{o}}(\hat{x}_k,x_k)\leq \overline{e}_k$ $\forall k\in\mathbb{I}_{[0,t-1]}$. 
Using Assumption~\ref{ass:disturbance}, the true state and disturbance sequence satisfy the constraints in~\eqref{eq:NLP_set_estimation} and thus $\hat{\gamma}_t\geq V_{\mathrm{o}}(\hat{x}_t,x_t)$. 
Satisfaction of Condition~\eqref{eq:thm_IOSS_1} follows from the update ~\eqref{eq:NLP_set_estimation_update} and Assumptions~\ref{ass:disturbance}, \ref{ass:stable_observer}. 
Conditions~\eqref{eq:thm_IOSS_2}--\eqref{eq:thm_IOSS_3} follow using the same arguments as in Theorem~\ref{thm:IOSS}. 
\end{proof}
The resulting bound $\overline{e}_t$ shares the same theoretical properties (cf.~\eqref{eq:thm_IOSS}) as the bounds used in Theorem~\ref{thm:IOSS}.   
\begin{remark}(Advantages and limitations)
\label{rk:set_estimation}
Compared to the bounds derived in Theorem~\ref{thm:IOSS}, the set estimation method used in Theorem~\ref{thm:set} has multiple advantages. 
%no additive disturbances
In particular, we can directly deal with more general disturbance characterizations (Ass.~\ref{ass:add_disturbance} is not needed). 
%no detectability
Furthermore, although detectability (cf.~Def.~\ref{def:IOSS}/\ref{def:IOSS_Lyap}) of the system is implicitly needed to ensure satisfaction of Assumption~\ref{ass:stable_observer},  the updates in Theorem~\ref{thm:set} do not use the corresponding constants or require an identical Lyapunov function (cf. Ass.~\ref{ass:common_Lyap}).
% less conservative
One of the main benefits of the set estimation (Thm.~\ref{thm:set}) is the fact that the exact nonlinear system equations are used to compute $\hat{\gamma}_t$, instead of using (possibly conservative) bounds $\overline{e}_{t,\mathrm{IOSS}}$ based on the $\delta$-IOSS  Lyapunov function (Def.~\ref{def:IOSS_Lyap}), thus resulting in less conservative estimates. 
  %
%Drawback
However, the set-membership approach also suffers from some inherent limitations.
%NLP
The update~\eqref{eq:NLP_set_estimation_update} requires the solution to the NLP~\eqref{eq:NLP_set_estimation} and thus significantly increases the computational complexity.
%local minima
In particular, this optimization problem is non-convex and the guarantees in Theorem~\ref{thm:set} only hold if the \textit{global} optimum is found (which is not necessarily required for the optimization problems appearing in MPC and MHE, cf.~\cite{scokaert1999suboptimal,schiller2020robust}). 
%
%Fragile
Furthermore, since set-membership methods use the exact model characterization, they can be \textit{fragile} to \textit{outlier} noise (in contrast to the $\delta$-IOSS bounds, cf. App.~\ref{app:IOSS}). 
In particular, if there exists a single disturbance realization $w_t$ which does not satisfy Assumption~\ref{ass:disturbance}, then the non-falsified set can be empty and the optimization problem~\eqref{eq:NLP_set_estimation} becomes infeasible.  
\end{remark}

\begin{remark}
\label{rk:set_estimation_M1}
(Existing set-valued state estimation for nonlinear systems)
We point out that there also exists a rich literature on set-valued state estimation that does not require the high computational cost of the update in Theorem~\ref{thm:set} (cf.~\cite{shamma1997approximate,rego2020guaranteed}). 
In particular, these methods only use the latest measurement ($M=1$) to update the set $\mathbb{E}_t$ by suitably over-approximating the nonlinear propagation $\mathcal{F}$~\eqref{eq:set_propagate} (e.g., using a local Taylor approximation or interval arithmetic). 
The resulting sets $\mathbb{E}_t$ can be parametrized without some nominal Luenberger observer, e.g., using constrained zonotopes (cf.~\cite{rego2020guaranteed}). 
However, it is not obvious how to ensure Conditions~\eqref{eq:thm_IOSS_2}--\eqref{eq:thm_IOSS_3} for these set-valued estimation methods. 
\end{remark}

%MHE
%!TEX root = ./Output.tex
%%%%%%%%%%%%%%%%%%%%%%%%%%%%%%%%%%%%%%%%%%%%%%%%%%%%%%%%%%%%%%%%%%%%%%%%%%%%%%%
\subsection{Moving horizon estimation}
\label{sec:estimation_MHE}
In Sections~\ref{sec:IOSS}--\ref{sec:estimation_setmember}, we derived valid upper bounds on the estimation error $\overline{e}_t$ for a state estimate resulting from a Luenberger-like observer~\eqref{eq:observer}. 
In the following, we show how to compute ``optimal'' state estimates $\hat{x}_t$ resulting in a smaller bound on the estimation error $\overline{e}_t$, by using an MHE scheme. 

In order to construct a simple arrival cost\footnote{%
In case a simple continuity bound $\sigma_\delta$ is not known, the arrival cost can be replaced by an initial state constraint $\hat{x}_{-M_t|t}=\hat{x}_{t-M_t}$. 
In this case, Condition~\eqref{eq:MHE_properties_3} needs to be replaced by a different bound, e.g., using additive disturbances (Ass.~\ref{ass:add_disturbance}) to create a suitable candidate solution (cf.~\cite[Thm.~3]{knufer2018robust}).} 
for the MHE scheme, we assume that the $\delta$-IOSS Lyapunov function $W_\delta$ is uniformly continuous. 
\begin{assumption}
\label{ass:IOSS_cont}
(Continuity $\delta$-IOSS Lyapunov function)
There exists a function $\sigma_\delta\in\mathcal{K}$, such that for any $x,\hat{x},\tilde{x}\in\mathbb{X}$, the $\delta$-IOSS Lyapunov function $W_\delta$ (Def.~\ref{def:IOSS_Lyap}) satisfies
\begin{align}
\label{eq:IOSS_cont}
|W_{\delta}(\hat{x},x)-W_{\delta}(\tilde{x},x)|\leq \sigma_\delta(\|\hat{x}-\tilde{x}\|).
\end{align}
\end{assumption}
We note that Condition~\eqref{eq:IOSS_cont} can be ensured on the compact set specified in Assumption~\ref{ass:boundedness} if $W_\delta$ is continuous, which can in turn be ensured by suitable continuity properties on the dynamics (cf.~\cite[Thm.~11]{allan2020detect}) .

At time $t$, the MHE scheme considers past input and output data $(u,y)$  in a window of length $M_t=\min\{t,M\}$, $M\in\mathbb{I}_{\geq 0}$, the past estimate $\hat{x}_{t-M_t}$, and solves the following NLP:
\begin{subequations}
\label{eq:MHE_IOSS}
\begin{align}
\label{eq:MHE_IOSS_cost}
\min_{\hat{w}_{\cdot|t},\hat{x}_{-M_t|t}}& \sum_{j=1}^{M_t}\eta^{j-1}\left(\sigma_1(\overline{w}+\|\hat{w}_{-j|t}\|)+\sigma_2(\|\hat{y}_{-j|t}-y_{t-j}\|) \right)\nonumber\\
&\eta^{M_t}\sigma_{\delta}(\|\hat{x}_{-M_t|t}-\hat{x}_{t-M_t}\|)\\
\label{eq:MHE_IOSS_1}
\text{s.t. }&\hat{x}_{j+1|t}=f_{\mathrm{w}}(\hat{x}_{j|t},u_{t+j},\hat{w}_{k|t}),~ j\in\mathbb{I}_{[-M_t,-1]},\\
\label{eq:MHE_IOSS_2}
&\hat{y}_{j|t}=h_{\mathrm{w}}(\hat{x}_{j|t},u_{t+j},\hat{w}_{j|t}),~ j\in\mathbb{I}_{[-M_t,-1]}.
\end{align}
\end{subequations}
We denote a minimizer to~\eqref{eq:MHE_IOSS} by $\hat{w}^*_{\cdot|t},\hat{x}^*_{-M_t|t}$ with the corresponding estimated state and output trajectory $\hat{x}^*_{\cdot|t}$, $\hat{y}^*_{\cdot|t}$. 
Note that in the cost we choose $\eta,\sigma_1,\sigma_2,\sigma_\delta,\eta$ based on the $\delta$-IOSS Lyapunov function $W_\delta$. 
A similar exponentially decaying cost has been previously suggested in~\cite{knufer2018robust}, compare also~\cite{knuefer2021MHE} for a for a more general asymptotic discounting. 
The MHE estimate is given by
\begin{subequations}
\label{eq:MHE_update}
\begin{align}
\label{eq:MHE_update_1}
\hat{x}_t:=&\hat{x}^*_{0|t},\\
\label{eq:MHE_update_2}
\overline{e}_{t}:=& \sum_{j=1}^{M_t}\eta^{j-1}\left(\sigma_1(\overline{w}+\|\hat{w}^*_{-j|t}\|)+\sigma_2(\|\hat{y}^*_{-j|t}-y_{t-j}\|) \right)\nonumber\\
&+\eta^{M_t}(\overline{e}_{t-M_t}+\sigma_\delta(\|\hat{x}^*_{-M_t|t}-\hat{x}_{t-M_t}\|)).
\end{align}
\end{subequations}
The following theorems summarizes the theoretical properties.  
\begin{theorem}
\label{thm:MHE}
Suppose the system admits an (exponential-decay) $\delta$-IOSS Lyapunov function (Def.~\ref{def:IOSS_Lyap}) $W_\delta$, Assumptions~\ref{ass:disturbance}, \ref{ass:boundedness}, and \ref{ass:IOSS_cont} hold, and  $W_{\delta}(\hat{x}_0,x_0)\leq \overline{e}_0$.
Then, there exists a function $\sigma_{\mathrm{f}}\in\mathcal{K}$ such that for any $t\in\mathbb{I}_{\geq 0}$, the estimates~\eqref{eq:MHE_update} satisfy
\begin{subequations}
\label{eq:MHE_properties}
\begin{align}
\label{eq:MHE_properties_1}
&W_{\delta}(\hat{x}_t,x_t)\leq \overline{e}_t,\\
\label{eq:MHE_properties_2}
&\|\hat{x}_{t+1}-f(\hat{x}_t,u_t)\|
\leq  \sigma_{\mathrm{f}}(\alpha_1^{-1}(\|\overline{e}_{t}\|)+\overline{w})+\alpha_1^{-1}(\|\overline{e}_{t+1}\|),\\
\label{eq:MHE_properties_3}
&\overline{e}_t\leq \dfrac{1-\eta^{M_t}}{1-\eta}\sigma_1(2\overline{w})+\eta^{M_t}(\overline{e}_{t-M_t}+\sigma_\delta(\alpha_1^{-1}(\overline{e}_{t-M_t}))).
\end{align}
\end{subequations}
\end{theorem}
\begin{proof}
\textbf{Part I: } Suppose that $W_\delta(\hat{x}_j,x_j)\leq \overline{e}_j$ for all $j\in\mathbb{I}_{[0,t-1]}$. 
Continuity (Ass.~\ref{ass:IOSS_cont}) ensures
\begin{align}
\label{eq:MHE_proof_1}
&W_\delta(\hat{x}^*_{-M_t|t},x_{t-M_t})\nonumber\\
\leq &\underbrace{W_\delta(\hat{x}_{t-M_t},x_{t-M_t})}_{\leq \overline{e}_{t-M_t}}+\sigma_{\delta}(\|\hat{x}^*_{-M_t|t}-\hat{x}_{t-M_t}\|).
\end{align}
Given that $\hat{x}^*,\hat{w}^*,\hat{y}^*$ is a trajectory of the system, we can use the same derivation as in Proposition~\ref{prop:IOSS} based on $\delta$-IOSS resulting in
\begin{align*}
&W_\delta(\hat{x}_t,x_t)\\
\stackrel{\eqref{eq:IOSS_Lyap_2}}{\leq} &\sum_{j=1}^{M_t} \eta^{j-1}\left(\sigma_1(\overline{w}+\|\hat{w}^*_{-j|t}\|)+\sigma_2(\|\hat{y}^*_{-j|t}-y_{t-j}\|) \right)\nonumber\\
&+\eta^{M_t} W_\delta(\hat{x}^*_{-M_t|t},x_{t-M_t})\stackrel{\eqref{eq:MHE_update_2},\eqref{eq:MHE_proof_1}}{\leq} \overline{e}_t. 
\end{align*}
Thus, $W_\delta(\hat{x}_t,x_t)\leq \overline{e}_t$ holds recursively for all $t\in\mathbb{I}_{\geq 0}$ using induction. \\
\textbf{Part II: } Boundedness of $x,\hat{x},u,w$ (Ass.~\ref{ass:boundedness}) and $f$ continuos ensures that there exists a function $\sigma_{\mathrm{f}}\in\mathcal{K}$ such that
\begin{align*}
&\|\hat{x}_{t+1}-f(\hat{x}_t,u_t)\|\leq\|x_{t+1}-f(\hat{x}_t,u_t)\|+\||x_{t+1}-\hat{x}_{t+1}\|\\
\leq & \sigma_{\mathrm{f}}(\||x_{t}-\hat{x}_{t}\|+\overline{w})+\||x_{t+1}-\hat{x}_{t+1}\|.
\end{align*}
Using $\alpha_1(\||x_{t}-\hat{x}_{t}\|)\stackrel{\eqref{eq:IOSS_Lyap_1}}{\leq}  W_\delta(\hat{x}_t,x_t)\leq \overline{e}_t$, we arrive at~\eqref{eq:MHE_properties_2}. \\
\textbf{Part III: }
A feasible candidate solution to~\eqref{eq:MHE_IOSS} is the true trajectory, i.e., $\hat{x}_{k|t}=x_{t+k}$, $\hat{w}_{k|t}=w_{t+k}$, $\hat{y}_{k|t}=y_{t+k}$, $k\in\mathbb{I}_{[-M_t,-1]}$. 
The initial estimate satisfies
\begin{align}
\label{eq:MHE_proof_2}
\alpha_1(\|\hat{x}_{-M_t|t}-\hat{x}_{t-M_t}\|)\stackrel{\eqref{eq:IOSS_Lyap_1}}{\leq} &W_\delta(\hat{x}_{t-M_t},\hat{x}_{-M_t|t})
\stackrel{\eqref{eq:MHE_properties_1}}{\leq} \overline{e}_{t-M_t}.
\end{align}
Given that this trajectory is a feasible candidate solution to~\eqref{eq:MHE_IOSS} and thus upper bounds the cost of the minimizer, we can obtain the following upper bound on $\overline{e}_t$
\begin{align*}
\overline{e}_t=&\sum_{j=1}^{M_t}\eta^{j-1}\left(\sigma_1(\overline{w}+\|\hat{w}^*_{-j|t}\|)+\sigma_2(\|\hat{y}^*_{-j|t}-y_{t-j}\|) \right)\\
&+\eta^{M_t}(\sigma_\delta(\|\hat{x}^*_{-M_t|t}-\hat{x}_{t-M_t}\|)+\overline{e}_{t-M_t})\nonumber\\
\leq &\sum_{j=1}^{M_t-1}\eta^{j-1}\sigma_1(\overline{w}+\|w_{t-j}\|)\\
&+\eta^{M_t}(\sigma_\delta(\|\hat{x}_{t-M_t}-x_{t-M_t}\|)+\overline{e}_{t-M_t})\\
\stackrel{\text{Ass.}~\ref{ass:disturbance},\eqref{eq:MHE_proof_2}}{\leq} &\dfrac{1-\eta^{M_t}}{1-\eta}\sigma_1(2\overline{w})+\eta^{M_t}(\overline{e}_{t-M_t}+\sigma_\delta(\alpha_1^{-1}(\overline{e}_{t-M_t}))). \qedhere
\end{align*}
\end{proof}
The overall theoretical properties derived for the MHE estimate contain the same qualitative features provided by the Luenberger-like observer (cf. Thm.~\ref{thm:IOSS}), which are needed for robust MPC.
%1.
First, Condition~\eqref{eq:MHE_properties_1} ensures that $\overline{e}_t$ is a valid upper bound on the estimation error that utilizes past measurements, analogous to~\eqref{eq:thm_IOSS_1}. 
%2.
Second, condition~\eqref{eq:MHE_properties_2} provides a bound on the difference between the nominal prediction model and the MHE estimate ``dynamics'', similar to Inequality~\eqref{eq:thm_IOSS_2}.
Here, we can see that the bound for the MHE is more complex and thus typically more conservative. 
%3.
Finally, we provided a formula (cf.~\eqref{eq:MHE_properties_3}) that allows for deterministic predictions of the future estimation error $\overline{e}_{t+k}$. 
Compared to the formula in~\eqref{eq:thm_IOSS_3} for the Luenberger-like observer, the MHE formulas are more complex and it is not immediately obvious that the MHE estimate is robustly stable.  
In case of full information estimation (FIE), i.e., $M_t=t$, the bounds~\eqref{eq:MHE_properties_1}, \eqref{eq:MHE_properties_3} directly provide a robust estimator with the Lyapunov function $W_\delta$.  
This is similar to the FIE analysis in~\cite{allan2019lyapunov} which used an additional stabilizability assumption instead of continuity of $W_\delta$ and requires additional terms to form a \textit{Lyapunov-like function} since no exponential discounting is used. 
In case a standard finite-horizon MHE is used, the bound~\eqref{eq:MHE_properties_3} can only ensure robust stability if additionally $\sigma_\delta\circ\alpha_1^{-1}$ is (locally) linearly bounded and a sufficiently large horizon $M$ is used.
This requirement is comparable to the linear/polynomial-exponential bounds used in~\cite[Thm.~9]{muller2017nonlinear},\cite[Thm.~1]{knufer2018robust}, \cite[Lemma~11]{schiller2020robust} and can be viewed as the MHE equivalent of the \textit{exponential cost controllability} used in the analysis of MPC without terminal constraints (cf.~\cite{grune2017nonlinear}, \cite[Ass.~3-4]{Koehler2020Regulation}). %grimm?
We point out that the requirement to use a long enough horizon $M$ can be relaxed if observability is assumed (cf.~\cite{michalska1995moving}), a more intricate arrival cost is designed (cf.~\cite{rao2003constrained}), or a stabilizing observer is integrated (cf.~\cite{schiller2020robust,gharbi2020proximity2}).
\begin{remark}(Constraints and suboptimality in MHE)
\label{rk:compare_MHE}
%constraints
One of the classical motivations of MHE is also the fact that a-priori knowledge of the system can be included in the MHE by using additional constraints. 
However, using additional constraints would not improve the resulting guarantees (unless detectability only holds on some constraint set, compare Remark~\ref{rk:constraints}).  
In case the system is additionally one-step controllable (Ass.~\ref{ass:add_disturbance}), then the absence of constraints allows us to treat a Luenberger estimate as a feasible candidate solution and thus inherit some of its stability properties (cf.~\cite{schiller2020robust}). 
Such an explicitly known candidate solution has the additional advantage that the resulting guarantees remain valid with suboptimal solutions, which is not the case with the bound~\eqref{eq:MHE_properties_3}. 
\end{remark}

  \begin{remark}
\label{rk:MHE_simple}
(Combined MHE-Luenberger)
While the MHE formulation can provide significantly less conservative bounds $\overline{e}_t$, this scheme also has significant drawbacks compared to the results in Theorem~\ref{thm:IOSS}.
The bound~\eqref{eq:MHE_properties_3} requires additional assumptions and a long enough horizon $M$ to ensure that $\overline{e}_t$ does not diverge.
Furthermore, the derived bound~\eqref{eq:MHE_properties_2} can be significantly more conservative compared to the relatively direct bound~\eqref{eq:thm_IOSS_2} for Luenberger-like observers. 
A simple way to keep the desired bounds~\eqref{eq:thm_IOSS_2}--\eqref{eq:thm_IOSS_3} is to use an additional case distinction verifying whether the MHE estimates computed at time $t+1$ satisfy~\eqref{eq:thm_IOSS_2}--\eqref{eq:thm_IOSS_3} with $k=1$. 
If this is not the case, we replace the update~\eqref{eq:MHE_update} with the one-step Luenberger update $\overline{e}_{t+1}=\eta\overline{e}_t+\sigma_4(\overline{w})$, $\hat{x}_{t+1}=\hat{f}(\hat{x}_t,u_t,y_t)$.
In this case, the bound~\eqref{eq:MHE_properties_3} is in general not valid. 
Nevertheless, since we expect the derived MHE bounds to be relatively conservative compared to the true performance, this simple case distinction allows us to (often) use the improved bound~\eqref{eq:MHE_update_2} from the MHE, while still using the bounds in~\eqref{eq:observer_prop} to predict valid bounds. 
Compared to, e.g.,~\cite{schiller2020robust,gharbi2020proximity2}, the proposed case distinction uses the Luenberger observer only as a back-up instead of incorporating it directly in the optimization problem.
\end{remark}

%!TEX root = ./Output.tex
%%%%%%%%%%%%%%%%%%%%%%%%%%%%%%%%%%%%%%%%%%%%%%%%%%%%%%%%%%%%%%%%%%%%%%%%%%%%%%%
\section{Robust output-feedback MPC}
\label{sec:tube}
In this section, we present the proposed output-feedback MPC schemes based on the estimation error bounds provided in Section~\ref{sec:estimation}.
First, we compute predictable bounds for a dynamic output-feedback tracking controller (Sec.~\ref{sec:tube_dynamic}). 
Then, we use these bounds to develop a tube-based output-feedback MPC that guarantees robust recursive feasibility and constraint satisfaction (Sec.~\ref{sec:tube_RPI}).
We also show how the MHE formulation in Section~\ref{sec:estimation_MHE} can be incorporated  to develop a joint robust MPC-MHE optimization problem.  
Finally, we show how the tube-based MPC formulation can be simplified to a constraint tightening, resulting in a computational demand comparable to nominal output-feedback MPC (Sec.~\ref{sec:tube_tight}). 
%
%!TEX root = ./Output.tex
%%%%%%%%%%%%%%%%%%%%%%%%%%%%%%%%%%%%%%%%%%%%%%%%%%%%%%%%%%%%%%%%%%%%%%%%%%%%%%%
\subsection{Tube dynamics for nonlinear output-feedback}
\label{sec:tube_dynamic}
In order to reduce the effect of uncertainty (disturbances, noise, estimation error), tube-based MPC schemes use an additional feedback to bound the deviation w.r.t. some nominal prediction. 
To this end, we need to analyse the joint incremental stability properties of the true state $x$, the estimated state $\hat{x}$, and some nominal prediction $\overline{x}$. 
In the linear case, the separation principle can be used to separately design a stable observer and tracking feedback to compute two Lyapunov functions/RPI sets (cf.~\cite{mayne2006robust}). 
We note that the conservatism of such a separate design can be reduced by computing one combined Lyapunov function/RPI set (cf.~\cite{kogel2017robust}). 
In the nonlinear case, stability of a state feedback in combination with a stable observer can be ensured assuming Lipschitz continuity of the involved functions and (local) exponential stability of the controller (cf.~\cite{magni2004stabilization}). 
Alternatively, the combination of an exponentially stable observer with an exponentially stabilizing feedback yields an exponentially stabilizing dynamic output feedback which can be constructed using contraction metrics (cf.~\cite{manchester2014output}). 
We study the combined closed loop using the concept of ISS (Def.~\ref{def:delta_ISS}/\ref{def:ISS_Lyap}) by interpreting the mismatch between the observer dynamics and the nominal dynamics as a disturbance. 
\begin{proposition}
\label{prop:RPI}
Let Assumptions~\ref{ass:disturbance}, \ref{ass:stable_observer}, \ref{ass:add_disturbance} hold.
Suppose the system admits an (exponential-decrease) $\delta$-ISS CLF (Def.~\ref{def:ISS_Lyap}). 
Then, there exist functions  $\gamma_{\mathrm{s,o}},\gamma_{\mathrm{s,w}}\in\mathcal{K}$ such that for any $x,\hat{x},\overline{x}\in\mathbb{X}$, $\overline{u}\in\mathbb{U}$ it holds
\begin{align}
\label{eq:RPI_ISS}
&V_{\delta}(f(\overline{x},\overline{u}),\hat{f}(\hat{x},u,y))\nonumber\\
\leq &\rho V_\delta(\overline{x},\hat{x})+\gamma_{\mathrm{s,o}}(V_{\mathrm{o}}(\hat{x},x))+\gamma_{\mathrm{s,w}}(\overline{w}),
\end{align}
with $u=\kappa(\hat{x},\overline{x},\overline{u})$, $y=h_{\mathrm{w}}(x,u,w)$. 
\end{proposition}
\begin{proof}
Analogous to Proposition~\ref{prop:IOSS}, we use the fact that additive disturbances (Ass.~\ref{ass:add_disturbance}) allow us write the observer dynamics as perturbed dynamics as follows:
\begin{align*}
\hat{f}(\hat{x},u,y)=f(\hat{x},u)+E_{\mathrm{x}}\hat{w}=f_{\mathrm{w}}(\hat{x},u,\hat{w}),
\end{align*}
with $\hat{w}=E_{\mathrm{x}}^\dagger \hat{L}(\hat{x},u,y)$. 
Analogous to Inequality~\eqref{eq:thm_IOSS_2}, Condition~\eqref{eq:observer_prop_3_mod} and Assumption~\ref{ass:disturbance} ensure
$\|\hat{w}\|\leq \|E_{\mathrm{x}}^\dagger\|(\gamma_{\mathrm{L},1}(V_{\mathrm{o}}(\hat{x},x))+\gamma_{\mathrm{L},2}(\overline{w}))$. 
The $\delta$-ISS CLF (Def.~\ref{def:ISS_Lyap}) applied with $u=\kappa(\hat{x},\overline{x},\overline{u})$, $w=0$, and $\tilde{w}=\hat{w}$ implies
\begin{align*}
&V_\delta(f(\overline{x},\overline{u}),\hat{f}(\hat{x},u,y))\\
\leq& \rho V_\delta(\overline{x},\hat{x})+\sigma_3(\|E_{\mathrm{x}}^\dagger\|(\gamma_{\mathrm{L},1}(V_{\mathrm{o}}(\hat{x},x))+\gamma_{\mathrm{L},2}(\overline{w})))\\
\leq& \rho V_\delta(\overline{x},\hat{x})+\gamma_{\mathrm{s,o}}(V_{\mathrm{o}}(\hat{x},x))+\gamma_{\mathrm{s,w}}(\overline{w}),
\end{align*}
with 
 $\gamma_{\mathrm{s,o}}:=\sigma_3\circ 2\|E_{\mathrm{x}}^\dagger\| \gamma_{\mathrm{L},1}\in\mathcal{K}$,
 $\gamma_{\mathrm{s,w}}:=\sigma_3\circ2\|E_{\mathrm{x}}^\dagger\| \gamma_{\mathrm{L},2}\in\mathcal{K}$, 
 where we used $\sigma(a+b)\leq \sigma(2a)+\sigma(2b)$ for any $a,b\geq 0$, $\sigma\in\mathcal{K}$.
\end{proof}
 Condition~\eqref{eq:RPI_ISS} characterizes the difference between the estimated state and some nominal state based on bounds on the observer error $V_{\mathrm{o}}$ and the disturbance magnitude $\overline{w}$. 
By combining this bound with the observer stability properties~\eqref{eq:observer_prop_2} we know that the corresponding dynamic output feedback ensures convergence to an RPI set characterized by: 
\begin{subequations}
 \label{eq:max_error_RPI}
 \begin{align}
V_{\mathrm{o}}(\hat{x},x)\leq \overline{e}_{\max}:=&\sigma_4(\overline{w})/(1-\tilde{\eta}),\\
V_{\delta}(\overline{x},\hat{x})\leq \overline{s}_{\max}:=&(\gamma_{\mathrm{s,o}}(\overline{e}_{\max})+\gamma_{s,\mathrm{w}}(\overline{w}))/(1-\rho).
\end{align}
\end{subequations}
These bounds are sufficient to plan a nominal trajectory with tightened constraints such that the true closed-loop satisfies the state and input constraints, similar to~\cite{mayne2006robust}.
However, these a priori bounds are unnecessarily conservative and we will overcome this conservatism by using a receding horizon MPC implementation. 
In particular, at each sampling time $V_{\delta}(\overline{x},\hat{x})$ is exactly measured and less conservative bounds on $V_{\mathrm{o}}(\hat{x},x)$ are obtained online using the bounds in Section~\ref{sec:estimation}.

%!TEX root = ./Output.tex
%%%%%%%%%%%%%%%%%%%%%%%%%%%%%%%%%%%%%%%%%%%%%%%%%%%%%%%%%%%%%%%%%%%%%%%%%%%%%%%
\subsection{Homothetic tube-based MPC}
\label{sec:tube_RPI}
In the following, we present a homothetic tube-based output-feedback MPC scheme, that combines the stability properties of the observer (Sec.~\ref{sec:estimation}) with the $\delta$-ISS CLF.

%cost
For the MPC formulation, we consider a continuous stage cost $\ell:\mathbb{X}\times\mathbb{U}\times\mathbb{R}_{\geq 0}\times\mathbb{R}_{\geq 0}\rightarrow\mathbb{R}$, a continuous terminal cost $V_{\mathrm{f}}:\mathbb{X}\times\mathbb{R}_{\geq 0}\times\mathbb{R}_{\geq 0}\rightarrow\mathbb{R}$, a terminal set $\mathbb{X}_{\mathrm{f}}\subseteq\mathbb{X}\times\mathbb{R}_{\geq 0}\times\mathbb{R}_{\geq 0}$, and a prediction horizon $N\in\mathbb{I}_{\geq 0}$.
The open-loop cost over the prediction horizon $N$ of a nominal predicted state and input trajectory $(\overline{x}_{\cdot|t}\in\mathbb{X}^{N+1}$, $\overline{u}_{\cdot|t}\in\mathbb{U}^{N}$) with associated bounds on the estimation error $\overline{e}_{\cdot|t}\in\mathbb{R}^{N+1}_{\geq 0}$, and on the tracking error w.r.t. the nominal trajectory $\overline{s}_{\cdot|t}\in\mathbb{R}^{N+1}_{\geq 0}$ is given by
\begin{align*}
&\mathcal{J}_N(\overline{x}_{\cdot|t},\overline{u}_{\cdot|t},\overline{e}_{\cdot|t},\overline{s}_{\cdot|t})\nonumber\\
:=&\sum_{k=0}^{N-1}\ell(\overline{x}_{k|t},\overline{u}_{k|t},\overline{e}_{k|t},\overline{s}_{k|t})+V_{\mathrm{f}}(\overline{x}_{N|t},\overline{e}_{N|t},\overline{s}_{N|t}).
\end{align*}
The dependence of the cost $\mathcal{J}_N$ on the error bounds $\overline{e},\overline{s}$ allows for the consideration of a worst-case stage cost, yielding robust performance guarantees (cf.~\cite[Rk.~5]{Koehler2020Robust}, \cite{bayer2014tube}). 
For simplicity of exposition, we consider state estimates based on the Luenberger-like observers~\eqref{eq:observer} with the simplifying bounds and assumptions in Corollary~\ref{corol:IOSS}. 
The overall algorithm can be readily adapted to use the other estimation bounds presented in Section~\ref{sec:estimation}. 
At time $t$, given the state estimate $\hat{x}_t,\overline{e}_t$ from~\eqref{eq:observer}, \eqref{eq:e_t_bound_3}, the output-feedback MPC is based on the following NLP:
\begin{subequations}
\label{eq:MPC_tube}
\begin{align}
&\min_{\overline{u}_{\cdot|t},\overline{x}_{0|t}} \mathcal{J}_N(\overline{x}_{\cdot|t},u_{\cdot|t},\overline{e}_{\cdot|t},\overline{s}_{\cdot|t})\\
\text{s.t. }
%initial constraint
\label{eq:MPC_tube_init}
& \overline{s}_{0|t}=V_\delta(\overline{x}_{0|t},\hat{x}_{t}),\\
%predicted dynamics
\label{eq:MPC_tube_x_pred}
&\overline{x}_{k+1|t}=f(\overline{x}_{k|t},u_{k|t}),~j\in\mathbb{I}_{[0,N-1]},\\
%bound observer
\label{eq:MPC_tube_e_pred}
&\overline{e}_{k|t}=\dfrac{1-\tilde{\eta}^k}{1-\tilde{\eta}}\sigma_4(\overline{w})+\tilde{\eta}^k\overline{e}_{t},~k\in\mathbb{I}_{[0,N]},\\
%bound tube
\label{eq:MPC_tube_s_pred}
&\overline{s}_{k+1|t}=\rho\overline{s}_{k|t}+\gamma_{\mathrm{s,o}}(\overline{e}_{k|t})+\gamma_{\mathrm{s,w}}(\overline{w}),\\
%constraints
\label{eq:MPC_tube_tight}
&(x_{k|t},\kappa(\hat{x}_{k|t},\overline{x}_{k|t},\overline{u}_{k|t}))\in\mathbb{Z},~k\in\mathbb{I}_{[0,N-1]},\\
&\forall x_{k|t},\hat{x}_{k|t}: V_{\delta}(\overline{x}_{k|t},\hat{x}_{k|t})\leq \overline{s}_{k|t},~V_{\mathrm{o}}(\hat{x}_{k|t},x_{k|t})\leq \overline{e}_{k|t},\nonumber\\
%terminal
\label{eq:MPC_tube_terminal}
&(\overline{x}_{N|t},\overline{e}_{N|t},\overline{s}_{N|t})\in\mathbb{X}_{\mathrm{f}}.
\end{align}
\end{subequations} 
We denote a minimizer to~\eqref{eq:MPC_tube} by $\overline{u}^*_{\cdot|t},\overline{x}^*_{0|t}$ with the corresponding state trajectory and error bounds $\overline{x}^*_{\cdot|t}$, $\overline{s}^*_{\cdot|t}$, $\overline{e}_{\cdot|t}^*$. 
The closed-loop operation is given by the observer~\eqref{eq:observer}, the updates~\eqref{eq:e_t_bound_3}, and the control law $u_t=\kappa(\hat{x}_t,\overline{x}^*_{0|t},\overline{u}_{0|t}^*)$. 
The considered output-feedback MPC formulation corresponds to a homothetic-tube formulation (cf.~\cite{rakovic2012homothetic,dong2020homothetic}) due to the variable scaling $\overline{e},\overline{s}$. 
A crucial feature of the constraint tightening is the fact that the conservatism $\overline{s}_{\cdot|t}$ is deterministically predicted based on the initial error bounds $\overline{s}_{0|t},\overline{e}_{t}$ and thus the updates in Section~\ref{sec:estimation} allow for a more aggressive, but safe, operation compared to the a priori bound~\eqref{eq:max_error_RPI}. 
We note that the tightened constraints~\eqref{eq:MPC_tube_tight} can be converted into simpler to implement sufficient conditions, given suitable continuity bounds on the constraints, compare Assumption~\ref{ass:constraints_cont} below. 
In the special case of exact state measurement ($\overline{e}_{k|t}=0$), the corresponding robust MPC framework unifies/generalizes existing robust MPC methods based on contraction metrics/incremental stability, which consider the special cases $\bar{s}_{k|t}=\bar{s}_{\max}$ (cf. rigid tube formulations~\cite{singh2019robust,bayer2013discrete}) or $\bar{s}_{0|t}=0$ (cf. constraint tightening formulations~\cite{kohler2018novel,Koehler2020Robust}, Sec.~\ref{sec:tube_tight}), respectively.

The overall offline and online computation are summarized in Algorithm \ref{alg:online} and \ref{alg:offline}, respectively.
\begin{algorithm}[H]
\caption{Online Computation}
\label{alg:online}
\begin{algorithmic}[1]
\State Update $\hat{x}_t$ and $\overline{e}_t$  using~\eqref{eq:observer} and~\eqref{eq:e_t_bound_3}.
 \Statex (Alternatives: MHE~\eqref{eq:MHE_IOSS}--\eqref{eq:MHE_update}; set-membership estimation~\eqref{eq:observer} and \eqref{eq:NLP_set_estimation}--\eqref{eq:NLP_set_estimation_update})
\State Solve the MPC optimization problem \eqref{eq:MPC_tube}.
\State Apply the control input: $u_t=\kappa(\hat{x}_t,\overline{x}^*_{0|t},\overline{u}_{0|t}^*)$.
\State Set $t=t+1$ and go back to 1.
\end{algorithmic}
\end{algorithm}
\begin{algorithm}[H]
\caption{Offline Computation}
\label{alg:offline}
\begin{algorithmic}[1]
\State Choose stage cost $\ell$, constraint set $\mathbb{Z}$, disturbance bound~$\bar{w}$.
\item Compute $\delta$-IOSS Lyapunov function $W_\delta$ (Def.~\ref{def:IOSS_Lyap}).
\State Design stable observer $\hat{f}$~\eqref{eq:observer} (Ass.~\ref{ass:stable_observer}).
\State Compute $\delta$-ISS CLF (Def.~\ref{def:ISS_Lyap}).
\State Compute terminal ingredients $V_{\mathrm{f}},\mathbb{X}_{\mathrm{f}}$ (Ass. \ref{ass:term}).
\end{algorithmic}
\end{algorithm}

In order to provide closed-loop properties, we also need to impose (standard) conditions regarding terminal ingredients.
\begin{assumption}
\label{ass:term} (Terminal ingredients)
There exists a control law $k_{\mathrm{f}}:\mathbb{X}\rightarrow\mathbb{U}$, such that for all $(\overline{x},\overline{e},\overline{s})\in\mathbb{X}_{\mathrm{f}}$ and all $x,\hat{x}\in\mathbb{X}$ satisfying $V_\delta(\overline{x},\hat{x})\leq \overline{s}$ and $V_{\mathrm{o}}(\hat{x},x)\leq \overline{e}$, it holds that
\begin{subequations}
\label{eq:term}
\begin{align}
\label{eq:term_1}
&(\overline{x}^+,\overline{e}^+,\overline{s}^+)\in\mathbb{X}_{\mathrm{f}},\\
\label{eq:term_2}
&(x,\kappa(\hat{x},\overline{x},\overline{u}))\in\mathbb{Z},\\
\label{eq:term_3}
&V_{\mathrm{f}}(\overline{x}^+,\overline{e}^+,\overline{s}^+)-V_{\mathrm{f}}(\overline{x},\overline{s},\overline{e})\nonumber\\
\leq& \ell(0,0,\bar{e}_{\max},\bar{s}_{\max})-\ell(\overline{x},\overline{u},\overline{e},\overline{s}).
\end{align}
\end{subequations} 
with $\overline{x}^+=f(\overline{x},\overline{u})$, $\overline{u}=k_{\mathrm{f}}(\overline{x})$,
and with any  $\overline{s}^+,\overline{e}^+\in\mathbb{R}_{\geq 0}$ satisfying $\overline{s}^+\leq\rho s+\gamma_{\mathrm{s,o}}(\overline{e})+\gamma_{\mathrm{s,w}}(\overline{w})$ and $\overline{e}^+\leq\tilde{\eta}\overline{e}+\sigma_4(\overline{w})$. 
\end{assumption}
A simple design satisfying Assumption~\ref{ass:term} is given by $\mathbb{X}_{\mathrm{f}}=\{(\overline{x},\overline{e},\overline{s})|~\overline{x}\in\overline{\mathbb{X}}_{\mathrm{f}},\overline{e}\leq \overline{e}_{\max},\overline{s}\leq \overline{s}_{\max}\}$, where $V_{\mathrm{f}}$, $\overline{\mathbb{X}}_{\mathrm{f}}$ are constructed using nominal design methods (cf., e.g.,~\cite[Sec.~2.5.5]{rawlings2017model}) and \eqref{eq:term_3} requires $\ell$ non-decreasing in $\overline{e},\overline{s}$ (cf.~\cite[Rk.~5]{Koehler2020Robust}, \cite{bayer2014tube}).  
In particular, a terminal equality constraint, i.e., $V_{\mathrm{f}}=0$, $\overline{\mathbb{X}}_{\mathrm{f}}=\{0\}$, is feasible for $\overline{e}_{\max}$ and $\overline{s}_{\max}$ sufficiently small, and thus for a sufficiently small disturbance bound $\overline{w}$ using~\eqref{eq:max_error_RPI}.
\begin{theorem}
\label{thm:tube}
Let Assumptions~\ref{ass:disturbance}, \ref{ass:stable_observer},  \ref{ass:add_disturbance}, \ref{ass:common_Lyap}, and \ref{ass:term} hold.
Suppose the system admits an (exponential-decay) $\delta$-IOSS Lyapunov function (Def.~\ref{def:IOSS_Lyap}) and an (exponential-decrease) $\delta$-ISS CLF (Def.~\ref{def:ISS_Lyap}).  
Suppose further that~\eqref{eq:MPC_tube} is feasible at $t=0$ and $V_{\mathrm{o}}(\hat{x}_0,x_0)\leq \overline{e}_0$. 
Then, for all $t\in\mathbb{I}_{\geq 0}$ the problem~\eqref{eq:MPC_tube} is feasible and the constraints are satisfied, i.e., $(x_t,u_t)\in\mathbb{Z}$,  for the closed loop resulting from Algorithm~\ref{alg:online}. 
If further $\ell$, $V_{\mathrm{f}}$ are non-decreasing in $\bar{s}$, $\bar{e}$ and $\mathbb{Z}$ is compact, then the following performance bound holds:
\begin{align}
\label{eq:performance_RMPC}
\limsup_{T\rightarrow\infty}\frac{1}{T}\sum_{t=0}^{T-1}\ell(\overline{x}^*_{0|t},\overline{u}^*_{0|t},\overline{e}^*_{0|t},\overline{s}^*_{0|t})\leq \ell(0,0,\overline{e}_{\max},\overline{s}_{\max}).
\end{align}
\end{theorem}
\begin{proof}
\textbf{Part I: }
Suppose that the optimization problem is feasible at some time $t\in\mathbb{I}_{\geq 0}$. 
First, note that the conditions from Corollary~\ref{corol:IOSS} hold and thus  Conditions~\eqref{eq:thm_IOSS_1} and \eqref{eq:thm_IOSS_3} imply $V_{\mathrm{o}}(\hat{x}_t,x_t)\leq \overline{e}_t$ and $\overline{e}_{t+1}\leq \overline{e}^*_{1|t}$.   
At time $t+1$, consider the standard candidate input sequence $\overline{u}_{k|t+1}=\overline{u}_{k+1|t}^*$, $k\in\mathbb{I}_{[0,N-1]}$, $\overline{u}_{N|t+1}=k_{\mathrm{f}}(\overline{x}_{N|t+1})$, $\overline{x}_{0|t+1}=\overline{x}^*_{1|t+1}$, with $\overline{x}_{\cdot|t+1},\overline{e}_{\cdot|t+1},\overline{s}_{\cdot|t+1}$ according to the dynamics~\eqref{eq:MPC_tube_init}--\eqref{eq:MPC_tube_s_pred}. 
Proposition~\ref{prop:RPI} ensures $\overline{s}_{0|t+1}\leq \overline{s}^*_{1|t}$.
Monotonicity of the error propagation~\eqref{eq:MPC_tube_e_pred}--\eqref{eq:MPC_tube_s_pred}
 ensures that $\overline{s}_{k|t+1}\leq \overline{s}^*_{k+1|t}$, $\overline{e}_{k|t+1}\leq \overline{e}^*_{k+1|t}$, $k\in\mathbb{I}_{[0,N-1]}$. 
For $k\in\mathbb{I}_{[0,N-2]}$, satisfaction of the tightened constraints~\eqref{eq:MPC_tube_tight} follows from feasibility of step $k+1$ at time $t$ and the following nestedness property
\begin{align*}
&\{(\hat{x},x)|~V_\delta(\overline{x}_{k|t+1},\hat{x})\leq\overline{s}_{k|t+1},~V_{\mathrm{o}}(\hat{x},x)\leq \overline{e}_{k|t+1}\}\\
\subseteq&\{(\hat{x},x)|~V_\delta(\overline{x}^*_{k+1|t},\hat{x})\leq \overline{s}^*_{k+1|t},~V_{\mathrm{o}}(\hat{x},x)\leq \overline{e}^*_{k|t+1}\}.
\end{align*}
 Condition~\eqref{eq:term_2} ensures that the constraint~\eqref{eq:MPC_tube_tight} is also feasible for $k=N-1$. 
Satisfaction of the terminal constraint~\eqref{eq:MPC_tube_terminal} follows from Assumption~\ref{ass:term}. 
Thus, the MPC problem is recursively feasible. \\
\textbf{Part II: }
Feasibility of constraint~\eqref{eq:MPC_tube_tight} with $k=0$, $V_{\mathrm{o}}(\hat{x}_t,x_t)\leq\overline{e}_t$, and $V_\delta(\overline{x}^*_{0|t},\hat{x}_t)\leq \overline{s}^*_{0|t}$ ensure that the true closed-loop state and input $x_t$, $u_t=\kappa(\hat{x}_t,\overline{x}^*_{0|t},\overline{u}^*_{0|t})$ satisfy the posed constraints $\mathbb{Z}$ for all $t\in\mathbb{I}_{\geq 0}$. \\
\textbf{Part III: }
Denote the value function corresponding to problem~\eqref{eq:MPC_tube} at time $t$ by $V_t$. 
Using standard arguments based on the candidate solution from Part I, $\ell,V_{\mathrm{f}}$ non-decreasing in $\bar{s},\bar{e}$, and condition~\eqref{eq:term_3} of the terminal cost $V_{\mathrm{f}}$ yields
\begin{align}
\label{eq:robust_performance_step}
V_{t+1}-V_t\leq \ell(0,0,\overline{e}_{\max},\overline{s}_{\max}) -\ell(\overline{x}^*_{0|t},\overline{u}^*_{0|t},\overline{e}^*_{0|t},\overline{s}^*_{0|t}).
\end{align}
Compact constraints $\mathbb{Z}$ in combination with the tightened constraint set~\eqref{eq:MPC_tube_tight} and the bounds~\eqref{eq:ISS_Lyap_1},\eqref{eq:observer_prop_1}, provide uniform bounds on the optimal solution of problem~\eqref{eq:MPC_tube}.  
Continuity of $\ell,V_{\mathrm{f}}$ then ensures boundedness of the value function $V_t$, $t\in\mathbb{I}_{\geq 0}$. % from $t=0$ till $t=T$ 
Summing up Inequality~\eqref{eq:robust_performance_step} and taking the average we arrive at~\eqref{eq:performance_RMPC}, analogous to the performance bound in~\cite{bayer2014tube}.
\end{proof}
By choosing $\ell$ as a worst-case stage cost (cf.~\cite[Rk.~5]{Koehler2020Robust}, \cite{bayer2014tube}) Inequality~\eqref{eq:performance_RMPC} ensures that the average closed-loop performance is no worse than the worst-case stationary performance around the origin.

\subsubsection*{Simultaneous control and estimation using MHE $\&$ MPC}
In the following, we briefly demonstrate how the MHE scheme from Section~\ref{sec:estimation_MHE} can be incorporated in the robust MPC formulation, resulting in simultaneous estimation and control with a single optimization problem.

In principle, it is possible to directly include the MHE optimization~\eqref{eq:MHE_IOSS} in the MPC~\eqref{eq:MPC_tube} by replacing the Luenberger estimate $\hat{x}_t,\overline{e}_t$ and the bounds~\eqref{eq:thm_IOSS_2}--\eqref{eq:thm_IOSS_3} with the MHE estimate and the bounds~\eqref{eq:MHE_properties_2}--\eqref{eq:MHE_properties_3}. 
However, due to the $M$-step nature of the bound~\eqref{eq:MHE_properties_3} and the possible conservatism of~\eqref{eq:MHE_properties_2}, the MHE-MPC may be less favourable compared to the Luenberger observer.
Thus, we instead use the combined Luenberger-MHE proposed in Remark~\ref{rk:MHE_simple}. 
We first assume that the MHE can satisfy conditions~\eqref{eq:thm_IOSS_2}--\eqref{eq:thm_IOSS_3} and use the Luenberger observer with the MPC~\eqref{eq:MPC_tube}, whenever the corresponding MHE-MPC problem does not guarantee the same performance bound (or even feasibility). 

The corresponding MHE-MPC optimization problem is given by the following NLP: 
\begin{subequations}
\label{eq:MPC_MHE_simple}
\begin{align}
&\min_{\overline{u}_{\cdot|t},\overline{x}_{0|t},\hat{x}_{-M|t},\hat{w}_{\cdot|t}} \mathcal{J}_N(\overline{x}_{\cdot|t},u_{\cdot|t},\overline{e}_{\cdot|t},\overline{s}_{\cdot|t})\\
\text{s.t. }
%MHE dynamics
&\eqref{eq:MHE_IOSS_1}\text{--}\eqref{eq:MHE_IOSS_2},\quad \eqref{eq:MPC_tube_x_pred}\text{--}\eqref{eq:MPC_tube_terminal},\\
%bound observer
\label{eq:MPC_MHE_simple_e_init}
&\overline{e}_{0|t}=\sum_{j=1}^{M_t}\eta^{j-1}(\sigma_1(\overline{w}+\|\hat{w}_{-j|t}\|)+\sigma_2(\|\hat{y}_{-j|t}-y_{t-j}\|))\nonumber\\
&+\eta^{M_t}(\sigma_\delta(\|\hat{x}_{-M|t}-\hat{x}_{t-M_t}\|)+\overline{e}_{t-M}),\\
%initial constraint
\label{eq:MPC_MHE_simple_init}
& \overline{s}_{0|t}=V_\delta(\overline{x}_{0|t},\hat{x}_{0|t}).
\end{align}
\end{subequations} 
The main difference to the optimization problem~\eqref{eq:MPC_tube} is that we additionally optimize over $\hat{x}_{-M|t},\hat{w}_{\cdot|t}$ to compute the  state estimate $\hat{x}_{0|t}$, while the estimation error $\bar{e}_{0|t}$ is computed in~\eqref{eq:MPC_MHE_simple_e_init} based on the result in Theorem~\ref{thm:MHE}. 
To ensure that the performance bound~\eqref{eq:performance_RMPC} remains valid, we use a case distinction based on 
\begin{align*}
\bar{V}_{t+1}:=V_{t}+\ell(0,0,\overline{e}_{\max},\overline{s}_{\max}) -\ell(\overline{x}^*_{0|t},\overline{u}^*_{0|t},\overline{e}^*_{0|t},\overline{s}^*_{0|t}),
\end{align*}
which represents a valid upper bound on the value function of the Luenberger observer based MPC (cf. \eqref{eq:robust_performance_step}), where the optimal solution at time $t$ corresponds to the solution utilized in the case distinction. 

The following algorithm summarizes the corresponding closed-loop operation. 
\begin{algorithm}[H]
\caption{MHE \& MPC}
\label{alg:MPC_MHE_simple}
\begin{algorithmic}[1]
\State Solve MHE-MPC optimization problem~\eqref{eq:MPC_MHE_simple} at time $t$.
\If{Problem~\eqref{eq:MPC_MHE_simple} is feasible \& $V_t\leq \bar{V}_t$ holds.}  
\State Set $\overline{e}_t=\overline{e}_{0|t}^*$, $\hat{x}_t=\hat{x}_{0|t}^*$.
\Else{ }
\State Use update~\eqref{eq:observer} and \eqref{eq:e_t_bound_3} and solve~\eqref{eq:MPC_tube}. 
\EndIf 
\State Apply the input: $u_t=\kappa(\hat{x}_t,\overline{x}^*_{0|t},\overline{u}_{0|t}^*)$ and compute $\bar{V}_{t+1}$.
\State Set $t=t+1$ and go back to 1.
\end{algorithmic}
\end{algorithm}
By combining the results in Theorems~\ref{thm:MHE} and~\ref{thm:tube}, this formulation guarantees robust constraint satisfaction and the robust performance bound~\eqref{eq:performance_RMPC}. 

%discuss related
Ideas to fuse MHE and MPC in a single optimization problem have also been suggested in~\cite{copp2017simultaneous} and \cite{dong2020homothetic}.
However, the approach in~\cite{copp2017simultaneous} uses a $\min-\max$ formulation, necessitating special solvers, and a weighted control/estimated objective was minimized, which complicates the analysis of closed-loop properties. 
In~\cite{dong2020homothetic}, the special case of linear systems is considered, which allows for an efficient incorporation of set-membership estimation similar to Section~\ref{sec:estimation_setmember} in terms of linear inequality constraints.

%!TEX root = ./Output.tex
%%%%%%%%%%%%%%%%%%%%%%%%%%%%%%%%%%%%%%%%%%%%%%%%%%%%%%%%%%%%%%%%%%%%%%%%%%%%%%%
\subsection{Simplified constraint tightening}
\label{sec:tube_tight}
In the following, we show how the tube-based MPC~\eqref{eq:MPC_tube} can be simplified to a nominal MPC with additional constraint tightening, similar to~\cite{kohler2018novel,Koehler2020Robust,Kohler2019Output}. 
The motivation for this reformulation is as follows:
The functions $V_\delta,\kappa,V_{\mathrm{o}}$ can have a highly nonlinear expression and thus the constraints~\eqref{eq:MPC_tube_init}, \eqref{eq:MPC_tube_tight} can significantly increase the computational complexity of the MPC.
For example, if CCMs are used, then evaluating $V_\delta,\kappa$ requires the computation of the geodesic, which is a nonlinear optimization problem (cf.~\cite{manchester2017control}). 

To provide a computationally efficient constraint tightening, we assume that the constraints are characterized by a set of continuous functions. 
\begin{assumption}
\label{ass:constraints_cont} (Continuous constraints) 
The constraint set is given by $\mathbb{Z}=\{(x,u)\in\mathbb{R}^{n+m}|~g_i(x,u)\leq 0,~i\in\mathbb{I}_{[1,r]}\}$ with $g_i$ continuous. 
Furthermore, there exist functions $\sigma_{\mathrm{g}_i,\mathrm{s}},\sigma_{\mathrm{g}_i,\mathrm{o}}\in\mathcal{K}$, $i\in\mathbb{I}_{[1,r]}$, with $\sigma_{\mathrm{g}_i,\mathrm{s}}$ superadditive  such that 
\begin{align}
\label{eq:constraints_cont}
g_i({x},\kappa(\hat{x},\overline{x},\overline{u}))-g_i(\overline{x},\overline{u})\leq \sigma_{\mathrm{g}_i,\mathrm{s}}(V_\delta(\overline{x},\hat{x}))+\sigma_{\mathrm{g}_i,\mathrm{o}}(V_{\mathrm{o}}(\hat{x},x)),
\end{align}
 for all $x,\hat{x}, \overline{x}\in\mathbb{X}$, $\overline{u}\in\mathbb{U}$, $i\in\mathbb{I}_{[1,r]}$. 
\end{assumption}
Condition~\eqref{eq:constraints_cont} follows from uniform continuity of $g_i$ and uniform bounds on $\kappa,V_\delta,V_{\mathrm{o}}$ (cf. Def.~\ref{def:IOSS_Lyap}, Ass.~\ref{ass:stable_observer}). 
Superadditivity of $\sigma_{\mathrm{g}_i,\mathrm{s}}$ allows for a more intuitive MPC formulation, compare the general conditions in~\cite[App.~B]{Koehler2020Robust}. 

Assumption~\ref{ass:constraints_cont} allows us to formulate the following MPC problem using tightened constraints: 
\begin{subequations}
\label{eq:MPC_tight}
\begin{align}
&\min_{\overline{u}_{\cdot|t}} \mathcal{J}_N(\overline{x}_{\cdot|t},u_{\cdot|t},\overline{e}_{\cdot|t},\overline{s}_{\cdot|t})\\
\text{s.t. }
%initial constraint
\label{eq:MPC_tight_init}
& \overline{s}_{0|t}=0,\quad 
\overline{x}_{0|t}=\hat{x}_t,\\
%predicted dynamics
&\eqref{eq:MPC_tube_x_pred}\text{--}\eqref{eq:MPC_tube_s_pred},\\ 
%constraints
\label{eq:MPC_tight_tight}
&g_i(\overline{x}_{k|t},\overline{u}_{k|t})+\sigma_{\mathrm{g}_i,\mathrm{s}}(\overline{s}_{k|t})+\sigma_{\mathrm{g}_i,\mathrm{o}}(\overline{e}_{k|t})\leq 0,\\
&i\in\mathbb{I}_{[1,r]},~k\in\mathbb{I}_{[0,N-1]},\nonumber\\
%terminal
\label{eq:MPC_tight_terminal}
&(\overline{x}_{N|t},\overline{e}_{t})\in\mathbb{X}_{\mathrm{f}}.
\end{align}
\end{subequations} 

Compared to the tube MPC~\eqref{eq:MPC_tube}, the nominal initial state is not treated as an optimization variable, but fixed to the measured state in~\eqref{eq:MPC_tight_init}, compare the difference between~\cite{mayne2006robust} and \cite{chisci2002feasibility}. 
Since $\overline{s}_{0|t}=0$ and $\overline{e}_{0|t}=\overline{e}_t$ are fixed at time $t$, the constraint tightening $\sigma_{\mathrm{g}_i,\mathrm{s}}(\overline{s}_{k|t})+\sigma_{\mathrm{g}_i,\mathrm{o}}(\overline{e}_{k|t})$ can be computed prior to solving the optimization problem, compare~\cite{Kohler2019Output}. 
Furthermore, since $\overline{s}_{0|t}$ is fixed, the terminal set constraint~\eqref{eq:MPC_tight_terminal} only needs to depend on the initial estimation error $\overline{e}_t$ in addition to the terminal state $\overline{x}_{N|t}$. 
Intuitively, we can think of Problem~\eqref{eq:MPC_tight} as a simplified version of Problem~\eqref{eq:MPC_tube} using $s_{0|t}=0$, which allows for a more computationally efficient implementation. 
The overall computational complexity of the optimization problem~\eqref{eq:MPC_tight} is equivalent to a nominal MPC (same number of constraints, decision variables). 
This simplified constraint tightening can equally be used to simplify the MHE-MPC formulation in~\eqref{eq:MPC_MHE_simple}. 

\begin{algorithm}[H]
\caption{Constraint tightening MPC}
\label{alg:online_4}
\begin{algorithmic}[1]
\State Update $\hat{x}_t$ and $\overline{e}_t$  using~\eqref{eq:observer} and~\eqref{eq:e_t_bound_3}.
\State Solve the MPC optimization problem \eqref{eq:MPC_tight}.
\State Apply the control input: $u_t=\overline{u}_{0|t}^*$.
\State Set $t=t+1$ and go back to 1.
\end{algorithmic}
\end{algorithm}
Compared to Algorithm~\ref{alg:online}, the optimized input $u_{0|t}^*$ is directly applied instead of using a stabilizing feedback $\kappa$. 

\begin{assumption}
\label{ass:term_tight}
(Terminal ingredients)
There exists a control law $k_{\mathrm{f}}:\mathbb{X}\rightarrow\mathbb{U}$, such that for all $(\overline{x},\overline{e})\in\mathbb{X}_{\mathrm{f}}$, and all $d_{\mathrm{w}}$ satisfying $V_\delta(\overline{x}^+ + d_{\mathrm{w}},\overline{x}^+)\leq \rho^N (\gamma_{\mathrm{s,o}}(\overline{e})+\gamma_{\mathrm{s,w}}(\overline{w}))$, 
it holds that
\begin{subequations}
\label{eq:term_tight}
\begin{align}
\label{eq:term_tight_1}
&(\overline{x}^+ +d_{\mathrm{w}},\overline{e}_1)\in\mathbb{X}_{\mathrm{f}},\\
\label{eq:term_tight_2}
&g_i(\overline{x},k_{\mathrm{f}}(\overline{x}))+\sigma_{\mathrm{g}_i,\mathrm{s}}(\overline{s}_N)+\sigma_{\mathrm{g}_i,\mathrm{o}}(\overline{e}_N)\leq 0,~i\in\mathbb{I}_{[1,r]},
\end{align}
\end{subequations} 
with $\overline{x}^+=f(\overline{x},\overline{u})$, $\overline{u}=k_{\mathrm{f}}(\overline{x})$, 
$\overline{e}_k=\dfrac{1-\tilde{\eta}^k}{1-\tilde{\eta}}\sigma_4(\overline{w})+\tilde{\eta}^k\overline{e}$, $\overline{s}_{k+1}=\rho\overline{s}_k+\gamma_{\mathrm{s,o}}(\overline{e}_k)+\gamma_{\mathrm{s,w}}(\overline{w})$, $k\in\mathbb{I}_{[0,N-1]}$, $\overline{s}_0=0$. 
\end{assumption}
Compared to Assumption~\ref{ass:term}, Condition~\eqref{eq:term_tight_1} requires an RPI condition and thus cannot be satisfied with a terminal equality constraint, compare~\cite[Prop.~6]{Kohler2019Output}, \cite[Prop.~4]{Koehler2020Robust} for a constructive offline design. 

\begin{theorem}
\label{thm:tight}
Let Assumptions~\ref{ass:disturbance}, \ref{ass:stable_observer},  \ref{ass:add_disturbance}, \ref{ass:common_Lyap}, \ref{ass:constraints_cont}, and \ref{ass:term_tight} hold.
Suppose the system admits an (exponential-decay) $\delta$-IOSS Lyapunov function (Def.~\ref{def:IOSS_Lyap}) and an (exponential-decrease) $\delta$-ISS CLF (Def.~\ref{def:ISS_Lyap}).  
Suppose further that~\eqref{eq:MPC_tight} is feasible at $t=0$ and $V_{\mathrm{o}}(\hat{x}_0,x_0)\leq \overline{e}_0$. 
Then, for all $t\in\mathbb{I}_{\geq 0}$, problem~\eqref{eq:MPC_tight} is feasible and the constraints are satisfied, i.e., $(x_t,u_t)\in\mathbb{Z}$, for the resulting closed loop. 
\end{theorem}
\begin{proof}
\textbf{Part I: }
Suppose that the optimization problem is feasible at some time $t\in\mathbb{I}_{\geq 0}$. 
Analogous to Theorem~\ref{thm:tube},  $\overline{e}_{k|t+1}\leq \overline{e}^*_{k+1|t}$, $k\in\mathbb{I}_{[0,N-1]}$. 
Denote $\overline{u}^*_{N|t}=k_{\mathrm{f}}(\overline{x}^*_{N|t})$, $\overline{x}^*_{N+1|t}=f(\overline{x}^*_{N|t},\overline{u}_{N|t}^*)$. 
As a candidate solution, consider $\overline{u}_{k|t+1}=\kappa(\overline{x}_{k|t+1},\overline{x}^*_{k+1|t},\overline{u}^*_{k+1|t})$, $k\in\mathbb{I}_{[0,N]}$.
Proposition~\ref{prop:RPI} with $w=0$, $\hat{x}=x$, and $\hat{f}(x,u,y)=f(x,u)$ ensures 
$V_\delta(\overline{x}_{k|t+1},\overline{x}^*_{k+1|t})\leq \rho^k V_\delta(\overline{x}_{0|t+1},\overline{x}^*_{1|t})\leq \rho^k \overline{s}^*_{1|t}$, $k\in\mathbb{I}_{[0,N]}$. 
Analogous to~\cite[Thm.~1]{Koehler2020Robust}, $\overline{e}^*_{k+1|t}\geq \overline{e}_{k|t+1}$ in combination with the formula~\eqref{eq:MPC_tube_s_pred} ensures $s^*_{k+1|t}\geq s_{k|t+1}+\rho^ks^*_{1|t}$. 
Superadditivity of $\sigma_{\mathrm{g}_i,\mathrm{s}}$ then implies
 \begin{align*}
\sigma_{\mathrm{g}_i,\mathrm{s}}(\overline{s}^*_{k+1|t})\geq \sigma_{\mathrm{g}_i,\mathrm{s}}(\overline{s}_{k|t+1})+\sigma_{\mathrm{g}_i,\mathrm{s}}(\rho^k\overline{s}^*_{1|t}),~k\in\mathbb{I}_{[0,N]}.
\end{align*}
Thus, the candidate solution satisfies the tightened constraints~\eqref{eq:MPC_tight_tight} for $k\in\mathbb{I}_{[0,N-2]}$ using
\begin{align*}
&g_i(\overline{x}_{k|t+1},\overline{u}_{k|t+1})+\sigma_{\mathrm{g}_i,\mathrm{s}}(\overline{s}_{k|t+1})+\sigma_{\mathrm{g}_i,\mathrm{o}}(\overline{e}_{k|t+1})\\
\stackrel{\eqref{eq:constraints_cont}}{\leq}& g_i(\overline{x}^*_{k+1|t},\overline{u}^*_{k+1|t})+\sigma_{\mathrm{g}_i,\mathrm{s}}(\rho^k \overline{s}^*_{1|t})\\
&+\sigma_{\mathrm{g}_i,\mathrm{s}}(\overline{s}_{k|t+1})+\sigma_{\mathrm{g}_i,\mathrm{o}}(\overline{e}_{k|t+1})\\
\leq& g_i(\overline{x}^*_{k+1|t},\overline{u}^*_{k+1|t})+\sigma_{\mathrm{g}_i,\mathrm{s}}(\overline{s}^*_{k+1|t})+\sigma_{\mathrm{g}_i,\mathrm{o}}(\overline{e}^*_{k+1|t})\stackrel{\eqref{eq:MPC_tight_tight}}{\leq}  0,
\end{align*}
where the first inequality used~\eqref{eq:constraints_cont} with $x=\hat{x}$. 
Satisfaction of the tightened constraints~\eqref{eq:MPC_tight_tight} for $k=N-1$ follows from Condition~\eqref{eq:term_tight_2} in Assumption~\ref{ass:term_tight}.
Note that $V_\delta(\overline{x}^*_{N+1|t},\overline{x}_{N|t+1})\leq \rho^Ns^*_{1|t}$ (cf. Part I), $s^*_{1|t}=\gamma_{s,o}(\overline{e}_t)+\gamma_{s,w}(\overline{w})$, and $(x^*_{N|t},\bar{e}_t)\in\mathbb{X}_{\mathrm{f}}$. 
Thus, Condition~\eqref{eq:term_tight_1} ensures satisfaction of the terminal constraint~\eqref{eq:MPC_tight_terminal}.
Hence, the MPC problem is recursively feasible. \\
\textbf{Part II: }
Feasibility of constraint~\eqref{eq:MPC_tight_tight} with $k=0$, $V_{\mathrm{o}}(\hat{x}_t,x_t)\leq\overline{e}_t$, $\hat{x}_t=\overline{x}_{0|t}$, $\overline{s}_{0|t}=0$, $u_t=\kappa(\bar{x}_{0|t},\bar{x}_{0|t},\bar{u}_{0|t})$ (cf.~\eqref{eq:ISS_Lyap_3}) in combination with Assumption~\ref{ass:constraints_cont} ensures that the true closed-loop state and input $x_t$, $u_t$ satisfy the posed constraints $\mathbb{Z}$ for all $t\in\mathbb{I}_{\geq 0}$. 
\end{proof}
%performance
Given a suitably defined terminal cost $V_{\mathrm{f}}$ (cf.~\eqref{eq:term_3}) and continuity bounds for $\ell,V_{\mathrm{f}}$, we can also show robust performance guarantees similar to~\eqref{eq:performance_RMPC} for this MPC formulation, compare~\cite[Thm.~1]{Koehler2020Robust} for a similar derivation. 
%exp. RMPC
The provided analysis also demonstrates that the previous robust MPC formulations in~\cite{kohler2018novel,Koehler2020Robust,Kohler2019Output}
 are \textit{not} restricted to incremental \textit{exponential} stability but can use exponential-decay $\delta$-CLFs, which allows for a broader range of applications.

%!TEX root = ./Output.tex
%%%%%%%%%%%%%%%%%%%%%%%%%%%%%%%%%%%%%%%%%%%%%%%%%%%%%%%%%%%%%%%%%%%%%%%%%%%%%%%
\section{Numerical example}
\label{sec:num} 
The following example considers a nonlinear robust output-feedback problem with state constraints and demonstrates the reduction in conservatism by using the MPC formulations in Section~\ref{sec:tube} in combination with the estimation error bounds from Section~\ref{sec:estimation}. 
In the following, the offline and online computation is done in Matlab using SeDuMi-1.3~\cite{sturm1999using} and CasADi\footnote{%
We use IPOPT and limit the number of iterations to $10^3$. In case the resulting solution is worse than the guarantees of the candidate solution (which can happen due to the difficult initialization), the candidate solution is used. 
}~\cite{andersson2019casadi}, respectively.
Furthermore, in the following we use the fact that the system properties only need to hold on the constraint set $\mathbb{Z}$ (cf. Rk.~\ref{rk:constraints}). 
%!TEX root = ./Output.tex
%%%%%%%%%%%%%%%%%%%%%%%%%%%%%%%%%%%%%%%%%%%%%%%%%%%%%%%%%%%%%%%%%%%%%%%%%%%%%%%
\subsubsection*{System model}
We consider the following 10-state quadrotor model 
\begin{align*}
\dot{z}_1=&v_1+{w}_1,\quad \dot{v}_1=g\tan(\phi_1)+w_6,\\
\dot{z}_2=&v_2+{w}_2,\quad \dot{v}_2=g\tan(\phi_2)+w_7,\\
\dot{z}_3=&v_3+w_3,\quad \dot{v}_3=-g+k_Tu_3+w_8,\\
\dot{\phi}_1=&-d_1\phi_1+\omega_1+w_4,\quad \dot{\omega}_1=-d_0\phi_1+n_0u_1+w_9,\\
\dot{\phi}_2=&-d_1\phi_2+\omega_2+w_5,\quad \dot{\omega}_2=-d_0\phi_2+n_0u_2+w_10,\\
x=&\begin{pmatrix}
z_1&z_2&z_3&\phi_1&\phi_2&v_1&v_2&v_3&\omega_1&\omega_2
\end{pmatrix}^\top\in\mathbb{R}^{10},\\
y=&\begin{pmatrix} I_5&0_5\end{pmatrix}x+\begin{pmatrix}0_{5\times 8}&5\cdot I_5\end{pmatrix}w\in\mathbb{R}^5,~ u\in\mathbb{R}^3, ~ w\in\mathbb{R}^{15},
\end{align*}
where $(z_1,z_2,z_3)$ are the positions, $(v_1,v_2,v_3)$ are the velocities, $(\phi_1,\phi_2)$ denote the pitch and roll angles, $(\omega_1,\omega_2)$ the pitch and roll rates, and $(u_1,u_2,u_3)$ are the adjustable pitch angle, roll angle and the vertical thrust. 
The parameters are $d_0=10,~d_1=8,~n_0=10,~k_T=0.91,~g=9.8$, and the constraint set is
\begin{align*}
\mathbb{Z}=\{(x,u)|~z_1\leq 4,~|\phi_{i}|\leq \pi/6,~|u_{1,2}|\leq \pi/9,~u_3\in[0,2g]\}.
\end{align*}
The discrete-time model is obtained using an Euler discretization with a sampling time $h=0.05$ and considering $u,w$ piece-wise constant. This essentially corresponds to the problem in~\cite{Koehler2020Robust}, with output measurements instead of state feedback and disturbance/noise on every variable. 
In the following, the noise and disturbances are chosen randomly such that $\|w_t\|= \bar{w}= 0.9\cdot 10^{-3}$.

%observer - detectability
For the offline design, we embed the differential dynamics in an LPV system (similar to a linear difference inclusions), thus allowing for the usage of simple LMI methods (cf.~\cite{koelewijn2021incremental,koehler2020nonlinearTAC,wang2020virtual,koelewijn2021nonlinear}). 
We verify detectability (Def.~\ref{def:IOSS_Lyap}) by computing a quadratic $\delta$-IOSS Lyapunov function $W_\delta$  resulting in $\eta=0.96$. %, $\sigma_1(r)=2.71\cdot r^2$, $\sigma_2(r)=8.54\cdot r^2$.   offline using LMI methods (cf. \cite{koelewijn2021incremental}),
We design a Luenberger observer of the form~\eqref{eq:obs_linear_gain}, which satisfies Assumption~\ref{ass:stable_observer} with $V_{\mathrm{o}}=W_\delta$.   
The corresponding design is chosen to achieve a faster nominal convergence rate of $\tilde{\eta}=0.957$. 
We compute a quadratic $\delta$-ISS CLF $V_\delta$ with a linear feedback $\kappa$ (Def.~\ref{def:ISS_Lyap}) satisfying~\eqref{eq:RPI_ISS} with $\rho=0.96$.  
%Given 
Given $W$ quadratic and $\max\{W(\hat{x},x),W(\tilde{x},x)\}\leq \overline{e}_{\max}$, Condition~\eqref{eq:IOSS_cont} (Ass.~\ref{ass:IOSS_cont}) holds with $\sigma_\delta(\|\hat{x}-\tilde{x}\|)=\|\hat{x}-\tilde{x}\|_P^2+2\|\hat{x}-\tilde{x}\|_P\sqrt{\overline{e}_{\max}}$, which is valid for the MHE since $W_\delta(\hat{x}_{t-M_t},x_{t-M_t})\leq \overline{e}_{t-M_t}\leq \overline{e}_{\max}$.

\begin{comment}
LDI-Prop:
\begin{align}
&V_{\delta}(f(\overline{x},\overline{u}),\hat{f}(\hat{x},u,y))\nonumber\\
=&\|(A_x+B_xK)(\bar{x}-\hat{x})+LC(\hat{x}-x)+LE_{\mathrm{y}}w)\|_P^2\\
\leq&\rho \|\bar{x}-\hat{x}\|_P^2+\gamma_{so}\|\hat{x}-x\|_{P_o}^2+\gamma_{sw}\|w\|^2
\end{align}
Write with $(\bar{x}-\hat{x},\hat{x}-x,w)=(z_1,z_2,z_3)$
\begin{align*}
&\begin{pmatrix}A_x+B_xK&LC&LE_{\mathrm{y}}\end{pmatrix}^\top
P_\delta
\begin{pmatrix}A_x+B_xK&LC&LE_{\mathrm{y}}\end{pmatrix}\\
\preceq&\text{diag}(\rho P_\delta,\gamma_{so}P_o,\gamma_{sw}I)
\end{align*}
$X=P_\delta^{-1}$, $K=YP$, $Y=KX$, multiple $(X,I,I)$
\begin{align*}
&\begin{pmatrix}A_xX+B_xY&LC&LE_{\mathrm{y}}\end{pmatrix}^\top
X^{-1}\begin{pmatrix}A_xX+B_xY&LC&LE_{\mathrm{y}}\end{pmatrix}\\
\preceq&\text{diag}(\rho X,\gamma_{so}P_o,\gamma_{sw}I)
\end{align*}
Schur
\begin{align*}
&\begin{pmatrix}X&\begin{pmatrix}A_xX+B_xY&LC&LE_{\mathrm{y}}\end{pmatrix}\\
\begin{pmatrix}A_xX+B_xY&LC&LE_{\mathrm{y}}\end{pmatrix}^\top&
\text{diag}(\rho X,\gamma_{so}P_o,\gamma_{sw}I)\end{pmatrix}\succeq 0
\end{align*}
Numerical values x1:
normal: 326, 6.6\\
Set E_y=0, then 70,0.9\\
E_x=0, then exact tracking possible? ($K=0$)
\end{comment}

\subsubsection*{Estimation}
First, we study only the estimation error bounds from Section~\ref{sec:estimation} with an exemplary trajectory, which is generated with the dynamic output-feedback from Proposition~\ref{prop:RPI}. 
For simplicity, we consider the initial condition $\hat{x}_0=x_0$, $\overline{e}_0=0$, i.e., perfect knowledge regarding the initial state.

We implemented the Luenberger observer and computed bounds $\overline{e}_t$ on the estimation error using: the a-priori bound~\eqref{eq:thm_IOSS_3}, the detectability bounds from Sections~\ref{sec:IOSS}, and the set-membership estimation from Section~\ref{sec:estimation_setmember} with horizon $M=4$. 
In addition, we implemented the MHE scheme from Section~\ref{sec:estimation_MHE} with horizon $M=10$. 
Furthermore, we also implement a set-membership estimation (cf. Sec.~\ref{sec:estimation_setmember}) to provide tighter bounds for the resulting MHE estimate. 
The results can be seen in Figure~\ref{fig:estimate}.

\begin{figure}[hbtp]
\begin{center}
\includegraphics[width=0.43\textwidth]{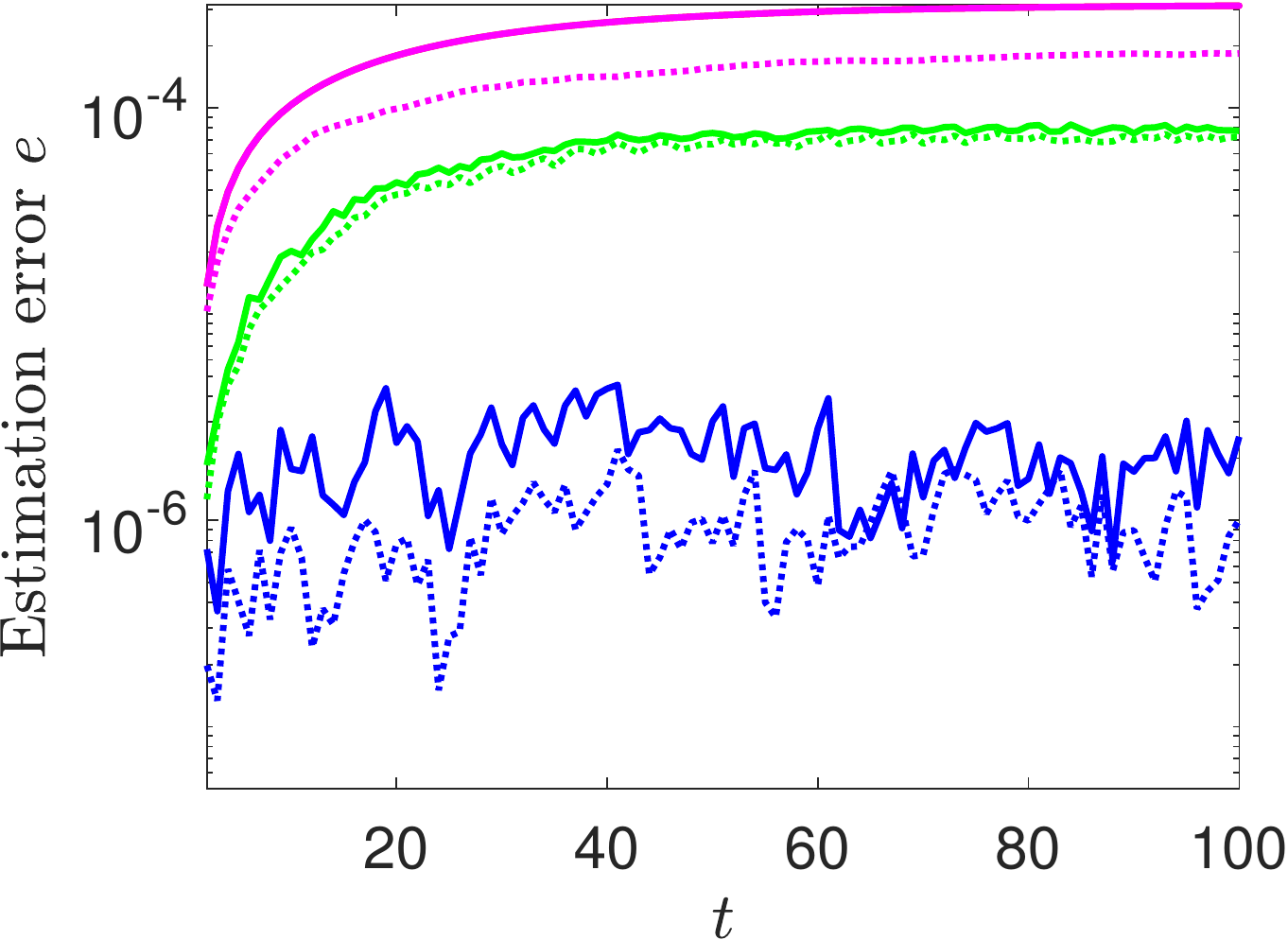}
\end{center}
\caption{Estimation error from the Luenberger observer is solid: true error $V_{\mathrm{o}}(\hat{x}_t,x_t)$ (blue), $\delta$-IOSS bound~\eqref{eq:e_t_bound_3} (magenta), set-membership estimation (green). 
Estimation error from the MHE is dotted: true error $V_{\mathrm{o}}(\hat{x}_t,x_t)$ (blue), $\delta$-IOSS bound~\eqref{eq:MHE_update} (magenta), set-membership estimation (green). 
}
\label{fig:estimate}
\end{figure}

%1. conservative
All the derived bounds are conservative over-approximations of the true estimation errors, typically with a factor of $10$ and higher.
This conservatism is most likely attributable to the general conservatism of worst-case robust bounds, especially given the high dimension of the disturbances $w\in\mathbb{R}^{15}$, which makes the occurrence of worst-case disturbances very unlikely.
%
%2. IOSS a priori - set-membership
Considering the Luenberger observer: 
The simple $\delta$-IOSS bound from Corollary~\ref{corol:IOSS} is easy to apply (scalar multiplication and min-operator), but it in the considered simulation it does not improve the a-priori error bounds. 
For the considered system we have $\sigma_4(\overline{w})/(1-\tilde{\eta})\approx 4\sigma_1(\overline{w})/(1-\eta)$. 
Thus, in the extreme case that no disturbances are encountered in closed-loop operation, i.e., $w_t=0$, the simple $\delta$-IOSS bound from Corollary~\ref{corol:IOSS} improves the a-priori bounds by $75\%$. 
The set-membership estimation (Sec.~\ref{sec:estimation_setmember} reduces the bounds on the estimation error on average by $75\%$. 
However, this comes at a significant increase in the online computational complexity.  
%Casadi:100ms
%
%3. MHE - set
Comparing the MHE and the Luenberger observer: the true estimation error of the MHE is on average only $36\%$ of the estimation error of the Luenberger observer, while the error bounds computed using set-membership estimation do not differ significantly. 
In addition, the online computed error bound of the MHE based on detectability is on average $40\%$ smaller than the a-priori error bound for the Luenberger observer. 
Although MHE also requires the solution to an NLP online, due to the different structure of the optimization problem, obtaining a solution to the MHE problem~\eqref{eq:MHE_IOSS} only required on average $5\%$ of the computational time compared to the set-membership estimation~\eqref{eq:NLP_set_estimation}.

As a summary: The simple $\delta$-IOSS bound~\eqref{eq:e_t_bound_3} is very easy to apply but seems to only provide improved estimates in case the assumed disturbance bound $\overline{w}$ is a conservative over-approximation of the true disturbances. 
The set-membership estimation method (Sec.~\ref{sec:estimation_setmember}) can provide significantly reduced bounds on the estimation error, however, the application also requires a globally optimal solution to a complex NLP and hence special care is required to avoid possible pitfalls. 
The MHE is able to obtain better state estimates including improved bounds on the estimation error compared to the a-priori bounds of a Luenberger observer, while the additional computational complexity is typically small compared to the MPC scheme.

\subsubsection*{Robust Output-feedback MPC}
In the following, we use the robust output-feedback MPC framework (Sec.~\ref{sec:tube}) and particularly focus on the effect of using different estimation methods. 
In order to demonstrate the reduction in conservatism, we wish to increase the position $x_1$ as much as possible while robustly guaranteeing that the constraint $x_1\leq 4$ is satisfied.
This is implemented with the following stage cost $\ell(\overline{x},\overline{u},\overline{e},\overline{s})=-\bar{x}_1+\ell_{\mathrm{s}}(\overline{s})+\ell_{\mathrm{e}}(\overline{e})$, which provides an upper bound on $-x_1$ for all $V_{\mathrm{o}}(\hat{x},x)\leq \overline{e}$, $V_\delta(\overline{x},\hat{x})\leq \overline{s}$, analogous to Assumption~\ref{ass:constraints_cont}.
Hence, the distance to the constraint can be viewed as a measure of the conservatism of the different robust formulation.  
We implemented the homothetic tube MPC~\eqref{eq:MPC_tube}, where \eqref{eq:MPC_tube_tight} is implemented using condition~\eqref{eq:constraints_cont}. 
Since $V_{\mathrm{s}}$ is quadratic, the additional complexity of the initial state constraint~\eqref{eq:MPC_tube_init} is limited and thus we do not need to implement the simplified constraint tightening from Section~\ref{sec:tube_tight}. 
We consider a prediction horizon of $N=40$ and a terminal equality constraint with the setpoint $x_{\mathrm{s}}=[3.7,3,10,0_7]^\top$, $u_{\mathrm{s}}=[0,0,g/k_T]^\top$, which satisfies Assumption~\ref{ass:term} with the smallest feasible stage cost $\ell(x_{\mathrm{s}},u_{\mathrm{s}},\overline{e}_{\max},\overline{s}_{\max})$. 
For simplicity, we consider the initial condition $x_0=\hat{x}_0=x_{\mathrm{s}}$, $\overline{e}_0=0$. 
We implemented the MPC using the Luenberger estimate with the a-priori error bound (Luen), the set-membership estimation (SetMember),  the MHE-Luenberger estimate (MHE) from Remark~\ref{rk:MHE_simple}),\footnote{%
Instead of verifying~\eqref{eq:thm_IOSS_2}--\eqref{eq:thm_IOSS_3} , we only need to verify that \eqref{eq:thm_IOSS_3} and \eqref{eq:RPI_ISS} remain valid, which is less conservative (and was satisfied in all closed-loop simulations). 
} and the joint MHE-MPC formulation (MHE-MPC) from Algorithm~\ref{alg:MPC_MHE_simple}. 
For comparison, we also implemented a simple rigid tube formulation (cf.~\cite{singh2019robust,bayer2013discrete}) with $\overline{s}=\overline{s}_{\max}$, $\overline{e}=\overline{e}_{\max}$.

\begin{figure}[hbtp]
\begin{center}
\includegraphics[width=0.43\textwidth]{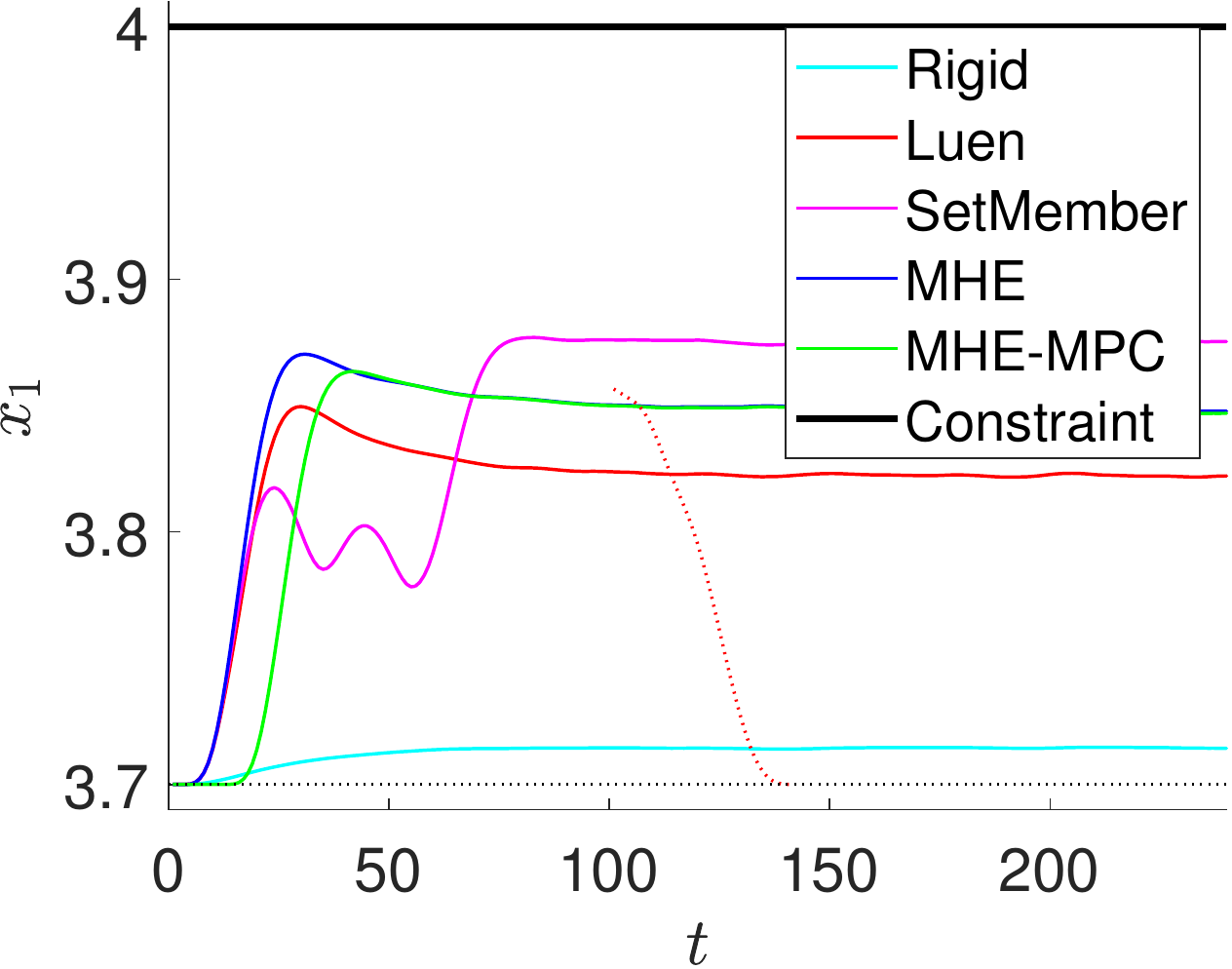}
\end{center}
\caption{Closed-loop position $x_1$: Rigid tube formulation (Rigid, cyan), homothetic tube formulation with the Luenberger observer based on the a-priori estimate (Luen, red) and set-membership estimation (SetMember, magenta); homothetic tube formulation based on the MHE estimate (MHE, blue) and the joint MHE-MPC formulation (MHE-MPC, green). State constraint $x_1\leq 4$ shown in black, solid and nominal steady-state $x_{\mathrm{s}}$ in black, dotted. 
Exemplary open-loop nominal trajectory $\overline{x}^*_{\cdot|t}$ at $t=100$ (Red, dotted).  
}
\label{fig:MPC}
\end{figure}

The results can be seen in Figure~\ref{fig:MPC}.
First, we can see that the smaller estimation error of using an MHE directly translates into a smaller conservatism of the output-feedback MPC and safe operation closer to the constraints. 
Regarding the joint MHE-MPC formulation (Alg.~\ref{alg:MPC_MHE_simple}), we see virtually no difference compared to the separate MHE and MPC formulation, except for a delay of $M=10$, which is due to the initialization of the MHE-MPC at $t=M$.  
In both MHE implementations, the case distinctions from Remark~\ref{rk:MHE_simple}/Algorithm~\ref{alg:MPC_MHE_simple} were never active and the Luenberger observer was not used. 
The set-membership estimation is able to provide even smaller error bounds resulting in further reduced conservatism, however, at the expense of a significant increase in the online computational demand. 
Lastly, the homothetic tube formulation has essentially the same computational complexity as the rigid tube formulation, however, the additional degrees of freedom in the homothetic tube formulation yield a significant improvement. 
From the exemplary open-loop trajectory in Figure~\ref{fig:MPC}, we can see that the nominal state predictions $\overline{x}$ differ drastically from the true closed-loop trajectories.    
Thus, overall we can see that the combination of the proposed output-feedback homothetic tube MPC with advanced state estimation methods (MHE, set-membership) outperforms competing approaches in terms of conservatism and performance.

%!TEX root = ./Output.tex
%%%%%%%%%%%%%%%%%%%%%%%%%%%%%%%%%%%%%%%%%%%%%%%%%%%%%%%%%%%%%%%%%%%%%%%%%%%%%%%
\section{Conclusion}
\label{sec:sum} 
We have presented a general framework for robust nonlinear output-feedback MPC.  
The main features of the proposed framework are the applicability to a rather general class of nonlinear constrained systems (incremental stabilizable \& detectable) and the fact that online bounds on the estimation error are incorporated to reduce conservatism. 
%
%
%Estimation
We have provided different methods to estimate the state and bound the observer error, with a varying degree of complexity and conservatism utilizing: a) stability of a Luenberger observer; b) detectability ($\delta$-IOSS); c) set-membership estimation; d) and "optimal" estimates using MHE. 
The proposed MPC formulation incorporates the resulting online bounds on the estimation error and ensures robust constraint satisfaction and performance. 
The corresponding MPC formulation also generalizes earlier nonlinear robust MPC methods~\cite{singh2019robust,bayer2013discrete,kohler2018novel,Koehler2020Robust} using a homothetic tube formulation and the corresponding computational complexity can be reduced to be equivalent to a nominal MPC scheme. 
%MHE_MPC
We showed how an MHE can be incorporated into the MPC formulation to yield a single optimization problem that jointly solves estimation and control, resulting in a larger feasible set. 
%example
We have demonstrated the improved performance of the proposed approach with a nonlinear numerical example. 
%
%Overall, the proposed framework is applicable to rather general class of nonlinear systems; can incorporate a variety of estimation methods; provides safety and performance guarantees; and avoids the conservatism of using worst-case estimation bounds.

%
%%%%%%%%%%%%%%%%%%%%%%%%%%%%%%%%%%%%%%%%%%%%%%%%%%%%%%%%%%%%%%%%%%%%%%%%%%%%%%%%
\bibliographystyle{IEEEtran}  
\bibliography{Literature}  

\begin{IEEEbiography}[{\includegraphics[width=1in,height=1.25in,clip,keepaspectratio]{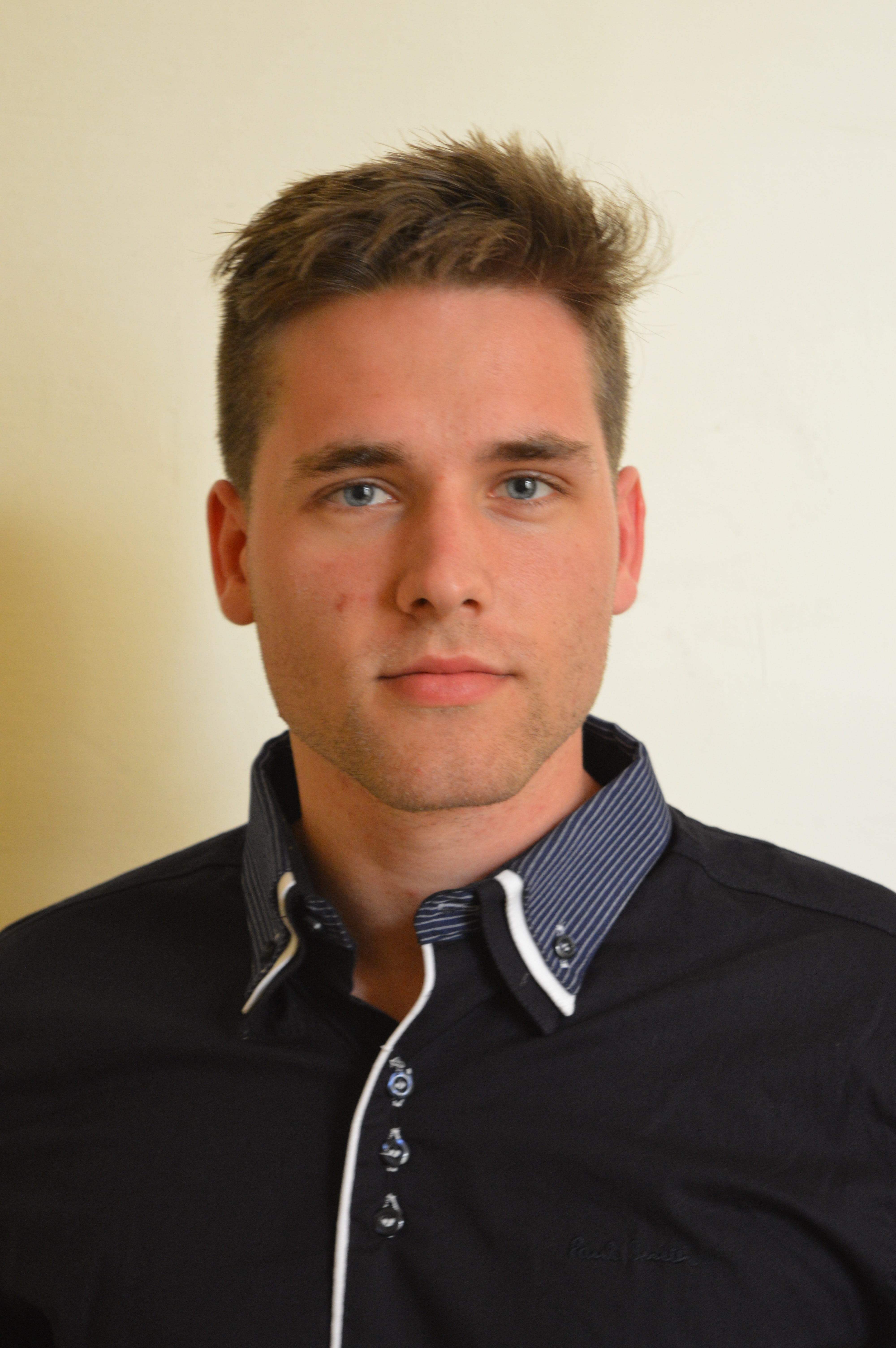}}]{Johannes K\"ohler}
 received his Master degree in Engineering Cybernetics from the University of Stuttgart, Germany, in 2017.
He has been a doctoral student at the \emph{Institute for Systems Theory and Automatic Control} at the University of Stuttgart under the supervision of Prof. Frank Allg\"ower from 2017 till 2021. 
Since then, he is a Postdoctoral researcher at the \textit{Institute for Dynamic Systems and Control} at ETH Zürich. 
His research interests are in the area of model predictive control and the control of nonlinear uncertain systems. 
\end{IEEEbiography}
\begin{IEEEbiography}[{\includegraphics[width=1in,height=1.25in,clip,keepaspectratio]{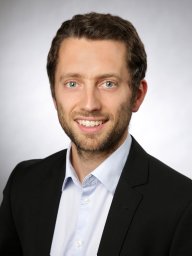}}]{Matthias A. M\"uller}
 received a Diploma degree in Engineering Cybernetics from the University of Stuttgart, Germany, and an M.S. in Electrical and Computer Engineering from the University of Illinois at Urbana-Champaign, US, both in 2009.
In 2014, he obtained a Ph.D. in Mechanical Engineering, also from the University of Stuttgart, Germany, for which he received the 2015 European Ph.D. award on control for complex and heterogeneous systems. Since 2019, he is director of the Institute of Automatic Control and full professor at the Leibniz University Hannover, Germany. 
He obtained an ERC Starting Grant in 2020 and is recipient of the inaugural Brockett-Willems Outstanding Paper Award for the best paper published in Systems \& Control Letters in the period 2014-2018. His research interests include nonlinear control and estimation, model predictive control, and data-/learning-based control, with application in different fields including biomedical engineering.
\end{IEEEbiography}
\begin{IEEEbiography}[{\includegraphics[width=1in,height=1.25in,clip,keepaspectratio]{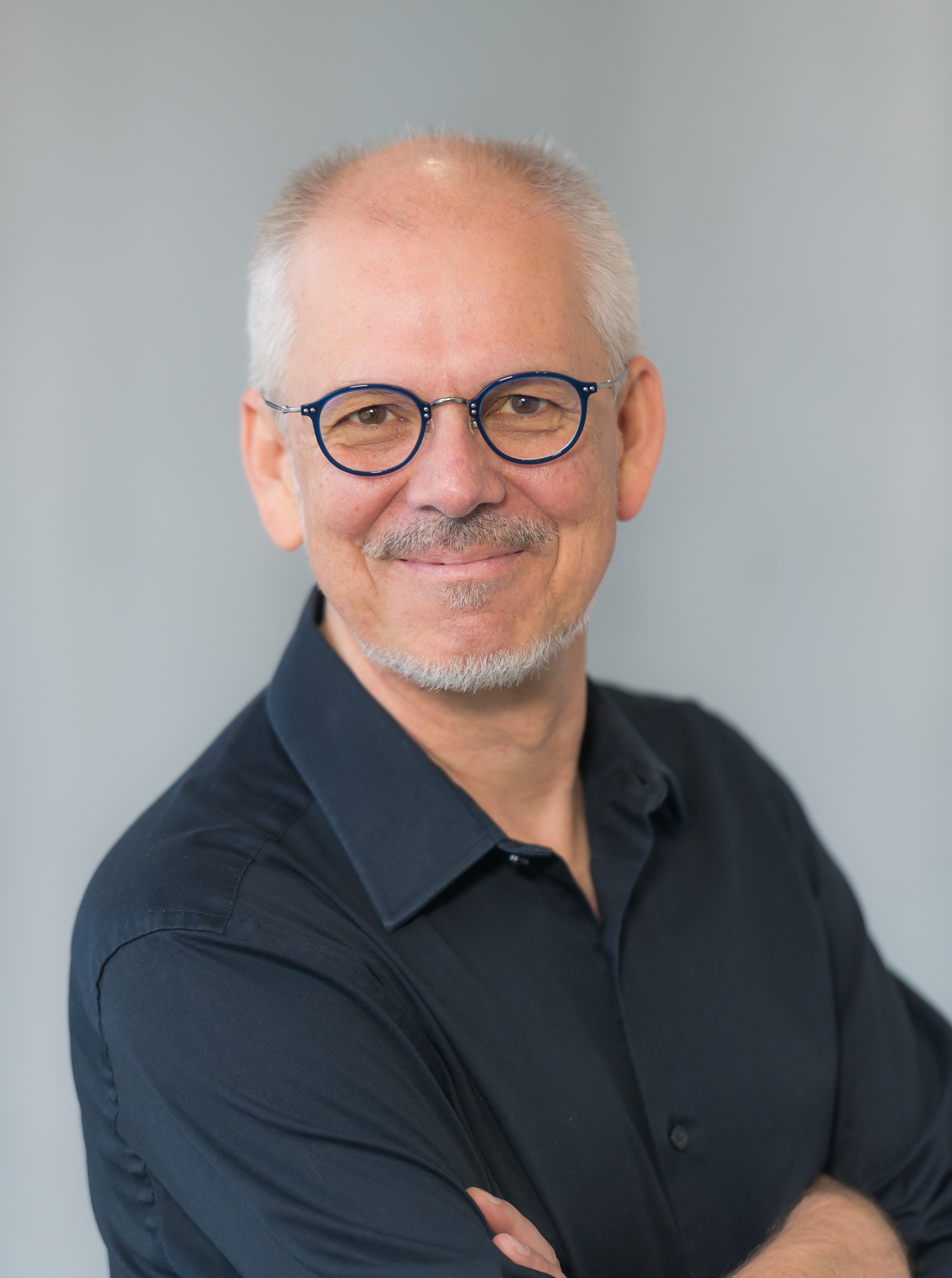}}]{Frank Allg\"ower}
is professor of mechanical engineering at the University of Stuttgart, Germany, and Director of the Institute for Systems Theory and Automatic Control (IST) there. 
Frank is active in serving the community in several roles: Among others he has been President of the International Federation of Automatic Control (IFAC) for the years 2017-2020, Vice-president for Technical Activities of the IEEE Control Systems Society for 2013/14, and Editor of the journal Automatica from 2001 until 2015. From 2012 until 2020 Frank served in addition as Vice-president for the German Research Foundation (DFG), which is Germany’s most important research funding organization. 
His research interests include predictive control, data-based control, networked control, cooperative control, and nonlinear control with application to a wide range of fields including systems biology.
\end{IEEEbiography}
%\newpage
 \appendix
%!TEX root = ./Output.tex
%%%%%%%%%%%%%%%%%%%%%%%%%%%%%%%%%%%%%%%%%%%%%%%%%%%%%%%%%%%%%%%%%%%%%%%%%%%%%
\subsection{Further estimation error bounds}
\label{app:IOSS}
In the following, we show that the $\delta$-IOSS bounds from Section~\ref{sec:IOSS} are not fragile to outlier noise and how observability can also be used to obtain estimation error bounds. 
\subsubsection*{Robustness to outlier noise}
In the following, we analyse the sensitivity and robustness of the bound $\overline{e}_t$ in Theorem~\ref{thm:IOSS}/Corollary~\ref{corol:IOSS} to outlier noise $w$ which does not satisfy Assumption~\ref{ass:disturbance}.
\begin{proposition}
\label{prop:IOSS_fragile}
Let Assumptions~\ref{ass:boundedness}, \ref{ass:stable_observer}, \ref{ass:add_disturbance}, and \ref{ass:common_Lyap}  hold.
There exist $\sigma_{\epsilon}\in\mathcal{K}$ and $\eta_\epsilon\in[0,1)$, such that 
for any $t\in\mathbb{I}_{\geq 0}$, the estimates $\hat{x}_t$, $\overline{e}_t$ according to~\eqref{eq:observer} and \eqref{eq:e_t_bound_3} satisfy 
\begin{align}
\label{eq:IOSS_fragile}
V_{\mathrm{o}}(\hat{x}_t,x_t)\leq& \overline{e}_t+\sum_{k=0}^{t-1}\eta_\epsilon^{t-k-1}\sigma_{\epsilon}(\max\{\|w_k\|-\overline{w},0\})\\
&+\eta_{\epsilon}^t\max\{V_{\mathrm{o}}(\hat{x}_0,x_0)-\overline{e}_0,0\}.\nonumber
\end{align}
\end{proposition}
\begin{proof}
Define $w_{t,\epsilon}:=\max\{\|w_t\|-\overline{w},0\}$ and $\epsilon_t:=\max\{V_{\mathrm{o}}(\hat{x}_t,x_t)-\overline{e}_t,0\}$, $t\in\mathbb{I}_{\geq 0}$. 
Given continuity of $\sigma_1,\sigma_4\in\mathcal{K}$, there exists a function $\sigma_{\epsilon}\in\mathcal{K}$ such that $\sigma_i(r_1+r_2)\leq \sigma_i(r_1)+\sigma_{\epsilon}(r_2)$, for all $i\in\{1,4\}$, and all $r_1,r_2 \in[0,c]$, with $c$ from Assumption~\ref{ass:boundedness}. 
Together with Assumption~\ref{ass:stable_observer}, this implies
\begin{align*}
V_{\mathrm{o}}(\hat{x}_{t+1},x_{t+1})\leq &\tilde{\eta}V_{\mathrm{o}}(\hat{x}_t,x_t)+\sigma_4(\overline{w}+w_{t,\epsilon})\\
\leq& \tilde{\eta}\overline{e}_t+\sigma_4(\overline{w})+\tilde{\eta}\epsilon_t+\sigma_\epsilon(w_{t,\epsilon}).
\end{align*}
The $\delta$-IOSS bound~\eqref{eq:observer_IOSS_bound_1} in Proposition~\ref{prop:IOSS}  remains valid with $\overline{w}$ replaced by $\overline{w}+w_{\epsilon,t-j}$. 
Using this bound for $M=1$, Equation~\eqref{eq:e_t_bound_3} and Assumption~\ref{ass:boundedness}, we arrive at
\begin{align*}
V_{\mathrm{o}}(\hat{x}_{t+1},x_{t+1})\leq \overline{e}_{t+1,\mathrm{IOSS}}+\eta \epsilon_{t}+\sigma_{\epsilon}(w_{t,\epsilon}). 
\end{align*}
Using the fact that $\min\{a+b,c+d\}\leq \min\{a,c\}+\max\{b,d\}$ and combining the above two inequalities yields  
\begin{align*}
&V_{\mathrm{o}}(\hat{x}_{t+1},x_{t+1})\\
\leq &\underbrace{\min\{\tilde{\eta} \overline{e}_t+\sigma_4(\overline{w}),\overline{e}_{t+1,\mathrm{IOSS}}\}}_{=\overline{e}_{t+1}}+\underbrace{\max\{\tilde{\eta},\eta\}}_{=:\eta_\epsilon}\epsilon_t+\sigma_{\epsilon}(w_{t,\epsilon}).\nonumber
\end{align*}
In case $\epsilon_{t+1}= V_{\mathrm{o}}(\hat{x}_{t+1},x_{t+1})-\overline{e}_{t+1}$, this implies $\epsilon_{t+1}\leq \eta_\epsilon\epsilon_t+\sigma_{\epsilon}(w_{t,\epsilon})$, which yields~\eqref{eq:IOSS_fragile} using the geometric series. 
In case $\epsilon_{t+1}=0$ the same bound holds since $\epsilon_t,w_{t,\epsilon}\geq 0$. 
\end{proof}
In case $V_{\mathrm{o}}(\hat{x}_0,x_0)\leq \overline{e}_0$ and $\|w_k\|\leq \overline{w}$ we recover the result in Corollary~\ref{corol:IOSS}. 
Inequality~\eqref{eq:IOSS_fragile} formalizes the sensitivity of the provided estimate $\overline{e}_t$ to outlier noise (Ass.~\ref{ass:disturbance} does not hold) and errors in the initial bound $\overline{e}_0$.

\subsubsection*{Observable systems}
In the following, we show how simpler bounds on the observer error $\overline{e}_t$ can be computed in case the system is observable. 
\begin{definition}
\label{def:observable}
(Final state observability~\cite[Def.~4.29]{rawlings2017model})
System~\eqref{eq:sys_w} is uniformly final state observable if there exist $\nu\in\mathbb{I}_{\geq 1}$, $\gamma_{\mathrm{w}},\gamma_{\mathrm{v}}\in\mathcal{K}$ such that
\begin{align}
\label{eq:observability}
&\|x_{t+\nu}-\tilde{x}_{t+\nu}\|\nonumber\\
\leq &\sum_{j=0}^{\nu-1}\gamma_{\mathrm{w}}(\|w_{t+j}-\tilde{w}_{t+k}\|)+\gamma_{\mathrm{v}}(\|y_{t+j}-\tilde{y}_{t+j}\|),
\end{align}
for all initial conditions $x_t,\tilde{x}_t\in\mathbb{X}$, all disturbance sequences $w,\tilde{w}\in\mathbb{W}^\infty$, all input sequences $u\in\mathbb{U}^{\infty}$, and all $t\in\mathbb{I}_{\geq 0}$, where $(x,u,y,w)_{t=0}^{\infty}$ and $(\tilde{x},u,\tilde{y},\tilde{w})_{t=0}^{\infty}$ correspond to two trajectories each satisfying Equations~\eqref{eq:sys_w} for all $t\in\mathbb{I}_{\geq 0}$.
\end{definition}
Final state observability is a stronger condition than detectability/$\delta$-IOSS, but weaker than observability (cf.~\cite[Prop.~4.31]{rawlings2017model}). 
Based on Definition~\ref{def:observable}, for $t\geq \nu$ the update Equations~\eqref{eq:e_t_bound_1}/\eqref{eq:e_t_bound_2} can be replaced by
\begin{subequations}
\label{eq:e_t_bound_obs}
 \begin{align}
\label{eq:e_t_bound_obs_a}
\overline{e}_{t,\mathrm{obs}}:=&\sum_{j=1}^{\nu}\gamma_{\mathrm{w}}(\overline{w}+\|\hat{w}_{t-j}\|)+\gamma_{\mathrm{v}}(\|y_{t-j}-\hat{y}_{t-j}\|),\\
\label{eq:e_t_bound_obs_b}
\overline{e}_t:=&\min\{\tilde{\eta} \overline{e}_{t-1}+\sigma_4(\overline{w}), \alpha_6(\overline{e}_{t,\mathrm{obs}})\}.
\end{align}
\end{subequations}
For $t< \nu$ we can simply define $\overline{e}_t=\dfrac{1-\tilde{\eta}^t}{1-\tilde{\eta}}\sigma_4(\overline{w})+\tilde{\eta}^t\overline{e}_0$.
The guarantees in Theorem~\ref{thm:IOSS} remain valid since $\overline{e}_{t,\mathrm{obs}}$ is an upper bound on $\|\hat{x}_{t}-x_{t}\|$ using Definition~\ref{def:observable}.

Final state observability is, e.g., satisfied for input-output models based on the extended (non-minimal) state $x_t=[u^\top_{[t-\nu,t-1]},y^\top_{[t-\nu,t-1]}]^\top$, which is often considered for data-driven control methods.
For such models, it is also possible to directly utilize the noisy measurements to define a state estimate $\hat{x}_t$ and bounds on the estimation error analogous to~\eqref{eq:thm_IOSS} can be directly deduced based on a noise bound (cf.~\cite{manzano2020robust}).

%!TEX root = ./Output.tex
%%%%%%%%%%%%%%%%%%%%%%%%%%%%%%%%%%%%%%%%%%%%%%%%%%%%%%%%%%%%%%%%%%%%%%%%%%%%%%%
\subsection{Identical Lyapunov function}
\label{app:identical}
In the following, we provide sufficient conditions, such that the $\delta$-Lyapunov function $V_{\mathrm{o}}$ (Ass.~\ref{ass:stable_observer}) is also a $\delta$-IOSS Lyapunov function, i.e., Assumption~\ref{ass:common_Lyap} holds.  
%This condition is  naturally satisfied for linear systems (cf.~\cite{cai2008input}) and, as the following proposition demonstrates, also for important special cases of nonlinear systems. 
\begin{proposition}
\label{prop:Lyap_is_IOSS}
Suppose that $f_{\mathrm{w}},h_{\mathrm{w}}$ are affine in $w$. 
Let Assumption~\ref{ass:stable_observer} hold with a quadratic function $V_{\mathrm{o}}$ and $\hat{L}$ satisfying Equation~\eqref{eq:obs_linear_gain}.
Then, Assumption~\ref{ass:common_Lyap} holds.
\end{proposition}
\begin{proof}
The proof follows the same lines as~\cite[Sec.~3.2]{cai2008input} for linear detectable systems.  
Given $V_{\mathrm{o}}$ quadratic with some positive definite matrix $P_{\mathrm{o}}\in\mathbb{R}^{n\times n}$ and $f_{\mathrm{w}}$, $h_{\mathrm{w}}$  affine in $w$ with matrices $E_{\mathrm{x}}\in\mathbb{R}^{n\times q}$, $E_{\mathrm{y}}\in\mathbb{R}^{p\times q}$, we have
\begin{align*}
&\|f_{\mathrm{w}}(x,u,w)-f_{\mathrm{w}}(\tilde{x},u,\tilde{w})\|_{P_{\mathrm{o}}}\\
\stackrel{\eqref{eq:obs_linear_gain}}{\leq} &\|\hat{f}(x,u,h(\tilde{x},u))-f(\tilde{x},u)\|_{P_{\mathrm{o}}}+\|E_{\mathrm{x}}(w-\tilde{w})\|_{P_{\mathrm{o}}}\\
&+\|L(h_{\mathrm{w}}(x,u,w)-h_{\mathrm{w}}(\tilde{x},u,\tilde{w}))-LE_{\mathrm{y}}(w-\tilde{w})\|_{P_{\mathrm{o}}}\\
\stackrel{\eqref{eq:observer_prop_2}}{\leq}&\sqrt{\tilde{\eta}}\|x-\tilde{x}\|_{P_{\mathrm{o}}}+\sqrt{\lambda_{\max}(E_{\mathrm{x}}^\top P_{\mathrm{o}} E_{\mathrm{x}})}\|w-\tilde{w}\|  \\
&+\sqrt{\lambda_{\max}(L^\top P_{\mathrm{o}} L)}\|h_{\mathrm{w}}(x,u,w)-h_{\mathrm{w}}(\tilde{x},u,\tilde{w})\|\\
&+\sqrt{\lambda_{\max}(E_{\mathrm{y}}^\top L^\top P_{\mathrm{o}}L E_{\mathrm{y}})}\|w-\tilde{w}\| .
\end{align*}
By squaring the result and using the fact that $\|a+b\|^2\leq (1+\epsilon)\|a\|^2+\frac{1+\epsilon}{\epsilon}\|b\|^2$ for any $\epsilon>0$, $a,b\in\mathbb{R}^n$, we arrive at Condition~\eqref{eq:IOSS_Lyap_2} with $\eta=(1+\epsilon)\tilde{\eta}$, $\sigma_1(r)=\dfrac{2(1+\epsilon)}{\epsilon}\left(\sqrt{\lambda_{\max}(E_{\mathrm{x}}^\top P_{\mathrm{o}} E_{\mathrm{x}})}+\sqrt{\lambda_{\max}(E_{\mathrm{y}}^\top L^\top P_{\mathrm{o}}L E_{\mathrm{y}})}\right)^2  r^2$, $\sigma_2=\dfrac{2(1+\epsilon)}{\epsilon}\lambda_{\max}(L^\top P_{\mathrm{o}}L) r^2$. 
For $\epsilon>0$ small enough we have  $\tilde{\eta}\in[0,1)$ and $\sigma_{1},\sigma_2\in\mathcal{K}$. 
Condition~\eqref{eq:IOSS_Lyap_1} holds with quadratic functions $\alpha_1=\alpha_5,\alpha_2=\alpha_6$  since $W_\delta=V_{\mathrm{o}}$ is quadratic and thus $V_{\mathrm{o}}$ is a $\delta$-IOSS Lyapunov function. 
\end{proof}

\end{document}